\documentclass{article}
\usepackage{graphicx,color}
\usepackage{amsmath,amssymb,amsthm}
\usepackage{framed,ascmac}
\usepackage[linesnumbered,ruled]{algorithm2e}
\usepackage{multicol,multirow}
\usepackage{cite}
\usepackage{tikz}
\usetikzlibrary{plotmarks}
\usetikzlibrary{arrows,positioning,automata,calc}
\usepackage{cleveref}
\usepackage{mathrsfs}
\usepackage[all]{xy}
\usepackage{comment}

\crefname{theorem}{Theorem}{Theorems}
\crefname{proposition}{Proposition}{Propositions}
\crefname{lemma}{Lemma}{Lemmas}
\crefname{exmp}{Example}{Examples}
\crefname{corollary}{Corollary}{Corollaries}
\crefname{claim}{Claim}{Claims}
\crefname{remark}{Remark}{Remarks}
\crefname{section}{Section}{Sections}
\crefname{definition}{Definition}{Definitions}
\crefname{example}{Example}{Examples}
\crefname{table}{Table}{Tables}
\crefname{appendix}{Appendix}{Appendices}
\crefname{equation}{Equation}{Equations}
\crefname{algorithm}{Algorithm}{Algorithms}
\crefname{subsection}{Subsection}{Subsections}

\newcommand{\cl}{{\rm cl}}
\newcommand{\conv}{{\rm conv}}

\SetKwInput{KwInput}{Input}         
\SetKwInput{KwOutput}{Output}

\newcommand{\Lnatural}{L$\sp{\natural}$}

\newcommand{\proj}{\pi}
\newcommand{\avgd}{\mu}
\newcommand{\median}{{\rm median}}
\newcommand{\aff}{{\rm aff}}

\renewcommand{\mid}{\,:\,}

\newcommand{\avg}[1]{\langle{#1}\rangle_{\avgd}}

\newcommand{\utvpi}{{\rm UTVPI}}

\newcommand{\be}{\subseteq}

\newcommand{\clq}[1]{{\rm cl}_{\bf{Q}}({#1})}

\newcommand{\douti}{\Leftrightarrow}

\newcommand{\nr}{\Rightarrow}

\newcommand{\bb}[1]{{\mathbb{#1}}}

\newcommand{\maj}{{\rm maj}}
\newcommand{\dom}{{\rm dom}}

\makeatletter
\def\iddots{\mathinner{\mkern1mu\raise\p@
		\hbox{.}\mkern2mu\raise4\p@\hbox{.}\mkern2mu
		\raise7\p@\vbox{\kern7\p@\hbox{.}}\mkern1mu}}
\makeatother

\usepackage[top=30truemm,bottom=30truemm,left=25truemm,right=25truemm]{geometry}

\newtheorem{theorem}{Theorem}[section]
\newtheorem{lemma}[theorem]{Lemma}
\newtheorem{proposition}[theorem]{Proposition}
\newtheorem{corollary}[theorem]{Corollary}
\newtheorem{definition}[theorem]{Definition}
\newtheorem{example}[theorem]{Example}
\newtheorem{remark}[theorem]{Remark}

\title{Characterizing the integer points in 2-decomposable polyhedra by closedness under operations}
\author{Kei Kimura, Kazuhisa Makino, Shota Yamada, and Ryo Yoshizumi}
\begin{document}
\maketitle

\begin{abstract}
Characterizing the solution sets in a problem by closedness under operations is recognized as one of the key aspects of algorithm development, especially in constraint satisfaction.
An example from the Boolean satisfiability problem is that the solution set of a Horn conjunctive normal form (CNF) is closed under the minimum operation, and this property implies that minimizing a nonnegative linear function over a Horn CNF can be done in polynomial time.
In this paper, we focus on the set of integer points (vectors) in a polyhedron, and study the relation between these sets and closedness under operations from the viewpoint of 2-decomposability.
By adding further conditions to the 2-decomposable polyhedra, we show that important classes of sets of integer vectors in polyhedra are characterized by 2-decomposability and closedness under certain operations, and in some classes, by closedness under operations alone.
The most prominent result we show is that the set of integer vectors in a unit-two-variable-per-inequality polyhedron can be characterized by closedness under the median and directed discrete midpoint operations, each of these operations was independently considered in constraint satisfaction and discrete convex analysis.
\end{abstract}

\section{Introduction}
\label{sec:introduction}

In this paper, we study the set of integer points (vectors) in a polyhedron, where 
a polyhedron is defined as a set of real vectors satisfying a system of linear inequalities.
Polyhedra have been intensively studied in mathematical optimization and computer science, where they arise as constraints in (integer) linear programming problems~\cite{Sch86,Van20}.
Since integer feasibility for a given polyhedron, i.e., determining whether there exists an integer vector in the polyhedron, is NP-hard, a special attention has been paid to polyhedra described by a system of linear inequalities of special form.
A typical example is polyhedra described by a difference constraint (DC) system (which we call DC polyhedra), where each inequality is of the form $x_i-x_j \ge c$ or $\pm x_i \ge c$ for some integer $c$.
Integer feasibility for DC polyhedra is known to be solvable in polynomial time by, e.g., the Bellman-Ford algorithm for the shortest path problem~\cite{Bel58,For56}.
Also, the set of integer vectors in a DC polyhedron coincides with a so called \Lnatural-convex set in discrete convex analysis and can be characterized by closedness under certain operations (see, e.g., \cite{Murota03} and \cref{thm:characterization-by-closure-property} in \cref{sec:pohyhedra-closure}).
For more examples, see Related Work below.
In this paper, we push the investigation of such a characterization further.

We focus on the set of integer vectors in a polyhedron, and study the relation between these sets and closedness under operations from the viewpoint of 2-decomposability.
2-decomposability is a concept proposed for the constraint satisfaction problem (CSP) studied in artificial intelligence and computer science, which is the computational problem of determining whether it is possible to assign values to variables while satisfying all given constraints.
Intuitively, we say that the set $S$ of integer vectors is 2-decomposable if, whenever each two-dimensional projection of a vector is in the corresponding projection of $S$, it is in $S$.
By definition, a 2-decomposable set of integer vectors can be written by constraints on two variables (although the converse does not hold in general).
Therefore, problems with 2-decomposability have attracted attention as an easy-to-handle class in the CSP.

2-decomposability has attracted attention for other reasons as well.
In particular, it has become clear that 2-decomposability is associated with closedness under operations, an important tool in analyzing the computational complexity of the CSP.
Intuitively, a set is closed under an operation when any vector obtained by applying the operation to multiple vectors in the set is also in the set.
Jeavons, Cohen, and Cooper showed that it is equivalent for all sets generated by applying certain relational operations from set $S$ to be 2-decomposable and for $S$ to be closed under a majority operation~\cite{JCC98} (see also \cref{thm:closed-near-unanimity<=>decompo} in \cref{sec:preliminaries}).
Here, a majority operation is a ternary operation that, whenever two arguments have the same value, then that value must be returned by the operation; in all other cases any value may be returned.

While the set of integer vectors in a polyhedron with 2-decomposability does not correspond exactly to the set of integer vectors closed under a majority operation, we show that by adding further conditions important classes of sets of integer vectors in polyhedra are characterized by 2-decomposability and closedness under certain operations, and in some classes, by closedness under certain operations alone.
The most prominent result we show is that the set of integer vectors in a unit-two-variable-per-inequality (UTVPI) polyhedron can be characterized by closedness under the median and directed discrete midpoint operations, each of these operations was independently considered in the CSP and discrete convex analysis.
Here, a polyhedron is called UTVPI if each inequality is of the form $\pm x_i \pm x_j \ge c$ or $\pm x_i \ge c$ for some integer $c$.
We have also clarified the relationship with the classes treated in discrete convex analysis such as integrally convex sets.

We further characterize 2-decomposability itself.
As noted above, in the known result in order to capture 2-decomposability by closedness under operations, it is necessary to assume 2-decomposability with respect to the infinite set generated by $S$. 
We observe that, in fact, closedness under operations cannot directly characterize 2-decomposability.
We then show that 2-decomposability can be characterized by introducing \emph{partial} operations and suitable closedness under such operations.

\subsection*{Related work}
In the CSP, the relationship between representation of the solution set and its closedness under operations is well studied, especially when the domain is finite.
In the finite domain Boolean case, the problem is a generalization of Boolean satisfiability problem (SAT), and the following results are known.
Let $n$ be a positive integer.
\begin{itemize}
\item A subset of $\{0,1\}^n$ is the solution set of a CNF where each clause has at least one negated 
 (resp., unnegated) variable if and only if it is closed under the zero (resp., one) operation.
\item A subset of $\{0,1\}^n$ is the solution set of a Horn (resp., dual-Horn) CNF if and only if it is closed under the minimum (resp., maximum) operation~\cite{McK43,Hor51}.
\item A subset of $\{0,1\}^n$ is the solution set of a 2-CNF if and only if it is closed under the majority operation~\cite{Sch78}.
\item A subset of $\{0,1\}^n$ is the solution set of a system of linear equations over GF(2) if and only if it is closed under the affine operation~\cite{Sch78}.
\end{itemize}
Here, a \emph{conjunctive normal form} (CNF) is conjunction of one or more clauses, a \emph{clause} is a disjunction of literals, and a \emph{literal} is an unnegated or negated variable.
A solution of a CNF is a truth assignment to the variables such that every clause has at least one true literal.
As is often done in the area of SAT, we identify true value with one and false value with zero.
The \emph{zero} (resp., \emph{one}) operation is a unary operation that always returns zero (resp., one).
A CNF is called \emph{Horn} (resp., \emph{dual-Horn}) if each clause of it has at most one unnegated (resp., negated) variable.
A CNF is called a \emph{2-CNF} if each clause of it has at most two variables.
The affine operation $\aff_2:\{0,1\}^3 \to \{0,1\}$ is define as $\aff_2(x,y,z) = x-y+z \pmod 2$.

In the finite domain non-Boolean case, the following results are known.
Let $D$ be a finite set.
\begin{itemize}
\item A subset of $D^2$ is 0/1/all if and only if it is closed under the dual discriminator operation~\cite{JCG95}.
\item When $D$ is a totally ordered set, a subset of $D^2$ is the solution set of a conjunctions of $(x_1 \ge d_1)\vee (x_2 \ge d_2)$, $(x_1 \ge d_1)\vee (x_2 \le d_2)$, $(x_1 \le d_1)\vee (x_2 \ge d_2)$, $(x_1 \le d_1)\vee (x_2 \le d_2)$, and $(x_i \le d_1)\vee (x_i \ge d_2)$ ($i=1,2$) if and only if it is closed under the median operation~\cite{JCC98,CJJ00}.
\item When $D$ is a totally ordered set, a subset of $D^n$ is the solution set of a conjunctions of the form $(x_1 < d_1)\vee (x_2 < d_2) \vee \dots \vee (x_r < d_r) \vee (x_i > d_i)$ if and only if it is closed under the minimum operation~\cite{JeC95}.
\item When $D=\{0,1,\dots, p-1\}$ for some prime $p$, a subset of $D^n$ is the solution set of a system of linear equations over GF($p$) if and only if it is closed under the affine operation~\cite{JCG95}.
\end{itemize}
Here, a subset $S$ of $D^2$ is called \emph{0/1/all} if, for each $d \in D$, $|\{ d'\in D \mid (d,d') \in S \}|$ equals either $0,1$, or $|\proj_2(S)|$, and symmetrically, for each $d \in D$, $|\{ d'\in D \mid (d',d) \in S \}|$ equals either $0,1$, or $|\proj_1(S)|$, where $\proj_i$ is a projection on the $i$th coordinate.
The dual discriminator operation $\theta$ is a ternary operation $\theta: D^3 \to D$ such that $\theta(d_1,d_2,d_3) =d_1$ if $d_1=d_2$ and $\theta(d_1,d_2,d_3) =d_3$ otherwise.
The median operation will be defined in \cref{sec:preliminaries}.
The affine operation $\aff_{p}:D^3 \to D$ is define as $\aff_p(x,y,z) = x-y+z \pmod p$.

In the infinite domain case, not much research has been done. 
It is worth noting that in the infinite case, no discussion about representation is possible without specifying how the constraints are given.
The following results are known.

\begin{itemize}
\item A semilinear set of $\bb{Q}^n$ is the solution set of semilinear Horn formula if and only if it is closed under the maximum operation~\cite{BoM18a}.
\item A semilinear set of $\bb{Q}^n$ is the solution set of semilinear tropically convex formula if and only if it is closed under the 
maximum operation and $+q$ operation for all $q \in \bb{Q}$~\cite{BoM18a}.
\item A semilinear set of $\bb{Q}^n$ is the solution set of a conjunction of bends if and only if it is closed under the median operation~\cite{BoM18b}.
\item A subset of $\bb{Z}^n$ is the solution set of a DC system if and only if it is closed under $g(x,y) = \lceil \frac{x+y}{2} \rceil$ and $h(x,y) = \lfloor \frac{x+y}{2} \rfloor$~\cite{Murota03}.
\item A subset of $\bb{Z}^n$ is the solution set of a system of linear inequalities of the form $x_i - x_j \ge c$ if and only if it is closed under the maximum, minimum, +1, and -1 operations~\cite{Murota03}.
\end{itemize}
Here, a subset of $\bb{Q}^n$ is \emph{semilinear} if it is a finite union of finite intersections of (open or closed) linear half spaces.
A \emph{semilinear Horn formula} is a finite conjunction of \emph{semilinear Horn clauses}, this is, finite disjunctions of the form $\bigvee_{i=1}^m(a^{(i)})^T x \succ_i c_i$ where $a^{(1)}, \dots, a^{(m)} \in \bb{Q}^n$ and there exists a $k \le n$ such that $a^{(i)}_j \ge 0$ for all $i$ and $j \neq k$, $\succ_i \in \{ \ge, > \}$, and $c_1, \dots c_m \in \bb{Q}$.
A \emph{semilinear tropically convex formula} is a finite conjunction of \emph{basic semilinear tropically convex formula}, this is, finite disjunctions of the form $\bigvee_{i=1}^m(a^{(i)})^T x \succ_i c_i$ where $a^{(1)}, \dots, a^{(m)} \in \bb{Q}^n$ and there exists a $k \le n$ such that $a^{(i)}_j \ge 0$ for all $i$ and $j \neq k$, $\succ_i \in \{ \ge, > \}$, $c_1, \dots c_m \in \bb{Q}$, and $\sum_{j}a^{(i)}_j = 0$ for all $i$.
For each $q \in \bb{Q}$, the $+q$ operation takes $x \in \bb{Q}$ as input and outputs $x+q$.
A \emph{bend} is a formula of the form $(x \circ_1 d_1) \vee (a_1x + a_2y \circ c) \vee (y \circ_2 d_2)$ where $\circ \in \{ \le, < \}$, $\circ_1, \circ_2 \in \{ \le, <, \ge, > \}$, $a_1,a_2 \in \bb{Q} \setminus \{0\}$, and $c,d_1,d_2 \in \bb{Q} \cup \{-\infty, +\infty \}$ are such that $\circ_i \in \{ \le, < \}$ if and only if $a_i > 0$ for $i =1,2$.

In summary, previous research has mainly been concerned with the representation of sets closed under a majority, the minimum (maximum), and the affine operations.

\subsection*{Outline}
The rest of the paper is organized as follows. 
\cref{sec:preliminaries} introduces the various definitions used in the paper.
In \cref{sec:2-decomp-closure}, we analyze the relationship between 2-decomposability and closure properties.
In \cref{sec:pohyhedra-closure}, we investigate the relationship between the sets of integer vectors in 2-decomposable polyhedra and closedness under operations.
In \cref{sec:2-decomp-characterization}, we characterize 2-decomposability by a certain closedness under partial operations.
Finally, we conclude the paper in \cref{sec:conclusion}.

\section{Preliminaries}
\label{sec:preliminaries}

Let $\mathbb{Z}$ and $\mathbb{R}$ denote the sets of integers and reals, respectively.
We assume that $n$ is an integer greater than one throughout the paper.
We also assume that $k$ is a positive integer, unless otherwise specified, throughout the paper.

We will analyze the relationship between polyhedral representations of a set of integer vectors and closedness under operations in terms of 2-decomposability.

First, we define 2-decomposability of a set $S \be \bb{Z}^n$.
For $\{ i_1,\dots, i_\ell \} \be \{1, \dots, n\}$ and a vector $x = (x_1,\dots, x_n) \in \bb{Z}^n$, the projection $\proj_{i_1,\dots, i_\ell}(x)$ is defined to be the $\ell$-dimensional vector $(x_{i_1},\dots, x_{i_\ell})$.
Moreover, for a set $S \be \bb{Z}^n$, 
let $\proj_{i_1,\dots, i_\ell}(S)=\{\proj_{i_1,\dots, i_\ell}(x) \mid x \in S \}$.
For $S_{i,j}\be \bb{Z}^2$ ($1\le i<j \le n$), the join $\Join_{1\le i<j\le n} S_{i,j} \be \bb{Z}^n$ is defined as 
$\Join_{1\le i<j\le n} S_{i,j}= \{ x\in \bb{Z}^n \mid \pi_{i,j}(x) \in S_{i,j} (1\le i<j\le n) \}$.

\begin{definition}[2-decomposability~\cite{JCC98}]
A set $S \subseteq \bb{Z}^n$ is said to be \emph{$2$-decomposable} if it contains all vectors $x$ such that $\proj_{i,j}(x) \in \proj_{i,j}(S)$ for all $\{i,j\} \be \{1,\dots, n\}$, or equivalently, $S=\Join_{1\le i<j\le n} \proj_{i,j}(S)$.
\end{definition}


For any ($k$-ary) operation $f:\mathbb{Z}^k \to \mathbb{Z}$ and any $x^{(1)},\dots , x^{(k)} \in \mathbb{Z}^n$, where $x^{(i)} = (x_1^{(i)},\dots , x_n^{(i)})$, define $f(x^{(1)},\dots , x^{(k)})$ to be $(f(x_1^{(1)},\dots, x_1^{(k)}),\dots , f(x_n^{(1)},\dots, x_n^{(k)}))\in \bb{Z}^n$.
Namely, $f(x^{(1)},\dots , x^{(k)})$ is obtained by a componentwise application of $f$ to $x^{(1)},\dots , x^{(k)}$.
We will characterize a subset $S \subseteq \mathbb{Z}^n$ in terms of closedness under operations.

\begin{definition}[Closedness under operation]
Let $S$ be a subset of $\mathbb{Z}^n$ and 
$f:\mathbb{Z}^k \to \mathbb{Z}$ be an operation.
We say $S$ is \emph{closed under} $f$, or \emph{$f$-closed} for short,  
if for all $x^{(1)},\dots , x^{(k)} \in S$, 
\begin{align*}
f(x^{(1)},\dots , x^{(k)}) \in S.
\end{align*}
\end{definition}

We particularly focus on the following operations.
The following one is introduced in~\cite{TaT21}.

\begin{definition}
The \emph{directed discrete midpoint} operation 
$\avgd: \mathbb{Z}^2 \rightarrow \mathbb{Z}$ is defined as follows.
For each $x,y \in \bb{Z}$, 

\begin{align*}
\avgd(x,y) = 
\begin{cases}
\lceil \frac{x+y}{2} \rceil & \text{if $x \ge y$}\\
\lfloor \frac{x+y}{2} \rfloor & \text{if $x < y$}.
\end{cases}
\end{align*}
Namely, if the average $\frac{x+y}{2}$ of $x$ and $y$ is an integer, then $\avgd$ outputs the integer, and otherwise, it outputs one of the two integers obtained by rounding $\frac{x+y}{2}$ that is ``close'' to $x$.
\end{definition}

\begin{definition}
A ternary operation $f: \mathbb{Z}^3 \rightarrow \mathbb{Z}$ is called a \emph{majority} operation, if for all $x,y \in \mathbb{Z}$, 
\begin{align*}
f(x,x,y) = f(x,y,x) = f(y,x,x) = x.
\end{align*}
\end{definition}

\begin{definition}
The median operation $\median: \mathbb{Z}^3 \rightarrow \mathbb{Z}$ is defined as follows.
Let $x,y,z \in \mathbb{Z}$.
Choose $u,v,w \in \mathbb{Z}$ such that 
$\{ u,v,w \} = \{ x,y,z \}$ and $u \le v \le w$.
Then $\median$ is defined by 
\begin{align*}
\median(x,y,z) = v.
\end{align*}
Namely, $\median$ outputs the median value of its three arguments.
\end{definition}

By definition, the median operation is a majority operation.

\begin{definition}
For a set $S\subseteq \mathbb{Z}^n$ and an operation $f:\mathbb{Z}^k \to \mathbb{Z}$, 
we define the 
\emph{$f$-closure} $\cl_f(S)$ of $S$ 
as an intersection of all sets that include $S$ and are $f$-closed.
\end{definition}

We remark that $\cl_f(S)$ itself includes $S$ and is closed under $f$. Hence, $\cl_f(S)$ is the (inclusion-wise) smallest set which includes $S$ and is closed under $f$. 
Moreover, $\cl_f$ is indeed a closure operator, namely, it is extensive (i.e., $S \be \cl_f(S)$), monotone (i.e., $S_1 \be S_2$ implies $\cl_f(S_1) \be \cl_f(S_2)$), and idempotent (i.e., $\cl_f(\cl_f(S)) = \cl_f(S)$).

\begin{definition}
For a set $S\subseteq \mathbb{Z}^n$, we define a \emph{convex closure} $\cl_{\conv}(S)$ of $S$ as an intersection of all convex sets that include $S$.
For a set $S\subseteq \mathbb{Z}^n$, we define a \emph{closed convex closure} $\cl_{\overline{\conv}}(S)$ of $S$ as an intersection of all (topologically) closed convex sets that include $S$.
\end{definition}

It is well known that $\cl_{\conv}$ and $\cl_{\overline{\conv}}$ are closure operators.

The following theorem shown in~\cite{JCC98} connects 2-decomposability with closedness under operations\footnote{\cref{thm:closed-near-unanimity<=>decompo} is originally shown for a set $S \be D^n$ for a finite set $D$ in \cite{JCC98}. However, in Subsection 4.4 in the same paper, it is mentioned that the result also holds for an infinite set $D$.}.
Several operations over the subsets of $\bb{Z}^n$ in the theorem are not defined here, since we only use the direction that (1) implies (2) in this paper.
\begin{theorem}
\label{thm:closed-near-unanimity<=>decompo}
For a set $S \be \bb{Z}^n$, the following conditions are equivalent:
\begin{itemize}
\item[(1)] $S$ is closed under a majority operation.
\item[(2)] Every $R$ in $S^+$ is 2-decomposable. In particular, $S$ is 2-decomposable.
\end{itemize}
Here, $S^+$ is the family of all subsets that can be obtained from $S$ using some sequence of 
(i) Cartesian product, (ii) equality selection, and (iii) projection operation.
\end{theorem}

Now, we provide some definitions for polyhedra.

\begin{definition}[Polyhedron]
A set $P \subseteq \mathbb{R}^n$ is called a \emph{polyhedron} if 
there exist $A \in \mathbb{R}^{m \times n}$ and $b \in \mathbb{R}^m$ such that 
$P=\{x\in \mathbb{R}^n \mid Ax \ge b \}$ holds.
\end{definition}

It is known that any polyhedron is a closed convex set.

\begin{definition}
Let $A\in \mathbb{R}^{m \times n}$ be a matrix.
\begin{itemize}
\item $A$ is said to be \emph{quadratic} if each row of A contains at most two nonzero elements.
\item $A$ is said to be \emph{linear} if each row of A contains at most one nonzero element.
\item $A$ is said to be \emph{unit} if each element of $A$ is either $0,-1$, or $ +1$.
\item $A$ is said to be \emph{difference} if each row of A contains at most one $+1$ and at most one $-1$.
\end{itemize}
\end{definition}

\begin{definition}
Let $P$ be a subset of $\mathbb{R}^n$.
\begin{itemize}
\item $P$ is called a \emph{TVPI (two-variable-per-inequality) polyhedron} if 
there exist quadratic matrix $A \in \mathbb{R}^{m \times n}$ and vector $b \in \mathbb{R}^m$ such that 
$P=\{x\in \mathbb{R}^n \mid Ax \ge b \}$ holds.
\item $P$ is called a \emph{UTVPI (unit-two-variable-per-inequality) polyhedron} if 
there exist unit quadratic matrix $A \in \{0,-1,+1\}^{m \times n}$ and vector $b \in \mathbb{R}^m$ such that 
$P=\{x\in \mathbb{R}^n \mid Ax \ge b \}$ holds.
\item $P$ is called a \emph{DC polyhedron} if 
there exist difference matrix $A \in \{0,-1,+1\}^{m \times n}$ and vector $b \in \mathbb{R}^m$ such that 
$P=\{x\in \mathbb{R}^n \mid Ax \ge b \}$ holds.
\item $P$ is called a \emph{SVPI (single-variable-per-inequality) polyhedron} if 
there exist linear matrix $A \in \{0,-1,+1\}^{m \times n}$ and vector $b \in \mathbb{R}^m$ such that 
$P=\{x\in \mathbb{R}^n \mid Ax \ge b \}$ holds.
\end{itemize}
\end{definition}

Note that a SVPI polyhedron is a DC polyhedron, a DC polyhedron is a UTVPI polyhedron, and a UTVPI polyhedron is a TVPI polyhedron.

\begin{definition}[Representability]
A set $S \subseteq \mathbb{Z}^n$ is said to be \emph{representable by a TVPI (resp, UTVPI, DC, SVPI) system} if $S = P \cap \mathbb{Z}^n$ for some TVPI (resp, UTVPI, DC, SVPI) polyhedron.
In this case, we say $P \cap \mathbb{Z}^n$ is a \emph{TVPI (resp, UTVPI, DC, SVPI) representation} of $S$.
\end{definition}

\begin{definition}
For a set $S\subseteq \bb{Z}^n$, 
we define a \emph{UTVPI closure} 
$\cl_{\utvpi}(S)$ 
as an intersection of all subsets of $\bb{Z}$ that include $S$ and are representable by UTVPI systems. 
\end{definition}

We remark that $\cl_{\utvpi}(S)$  itself includes $S$ and is representable by a UTVPI system. 
Hence, $\cl_{\utvpi}(S)$  is the (inclusion-wise) smallest set which includes $S$ and is representable by a UTVPI system. 
Indeed, $\cl_{\utvpi}$ is a closure operator.

\begin{proposition}[{\cite[Theorem 3.2]{BoM18b}}]
\label{prop:TVPI=>median-closed}
Let $S$ be a subset of $\mathbb{Z}^n$.
If $S$ is representable by a TVPI system, then $S$ is closed under the median operation.
\end{proposition}
\begin{remark}
\label{rem:infiniteTVPI=>median-closed}
\cref{prop:TVPI=>median-closed} can be easily generalized to the case where $S$ is representable by an infinite set of two-variable inequalities, i.e., $S = \{ x\in \mathbb{Z}^n \mid a_{i1}x_{j_i} + a_{i2}x_{k_i} \ge b_i (\forall i \in I) \}$ for an infinite set $I$, where $a_{i1},a_{i2},b_i \in \mathbb{R}$ for each $i \in I$.
\end{remark}

We now introduce notions from discrete convex analysis.

\begin{definition}[Integer neighborhood]
For vector $x \in \mathbb{R}^n$, 
its \emph{integer neighborhood} $N(x) \subseteq \mathbb{Z}^n$ is defined as 
$N(x)=\{ z \in \mathbb{Z}^n \mid |z_j -x_j| < 1\ (j=1,\dots, n) \}$.
\end{definition}

\begin{definition}[Midpoint-neighbor-closedness]
A set $S \subseteq \mathbb{Z}^n$ is said to be \emph{midpoint-neighbor-closed} 
if for each $x,y \in S$, the integer neighborhood of $\frac{x+y}{2}$ is contained in $S$, 
i.e. $N(\frac{x+y}{2}) \subseteq S$.
\end{definition}

\begin{definition}[Integral convexity]
A set $S \subseteq \mathbb{Z}^n$ is said to be \emph{integrally convex} 
if the convex closure of $S$ coincides with the union of 
the convex closures of $S \cap N(x)$ over $x \in \mathbb{R}^n$,
i.e., if, for any $x \in \mathbb{R}^n$, $x \in \cl_{\conv}(S)$ implies $x \in \cl_{\conv}(S \cap N(x))$.
\end{definition}

\begin{theorem}[Theorem 2.1 in~\cite{MT23}]
\label{thm:IntConvSet-characterization}
A set $S \be \bb{Z}^n$ is integrally convex if and only if 
\begin{align}
\frac{x+y}{2} \in \cl_{\conv}\left(S\cap N(\frac{x+y}{2})\right)
\end{align}
for every $x,y \in S$ with $||x-y||_{\infty} \ge 2$. 
\end{theorem}

\begin{proposition}[Proposition 2.1 in~\cite{MMTT19}]
\label{prop:two-dim-IntConv=UTVPI}
A set $S \subseteq \mathbb{Z}^2$ is an integrally convex set if and only if $S$ is representable by a UTVPI system.
\end{proposition}

\section{2-decomposability and closure properties}
\label{sec:2-decomp-closure}

In this section, we investigate the relationships between 2-decomposability and closure properties, which will be used in \cref{sec:pohyhedra-closure}.

\begin{lemma}
\label{lemma:proj-closed}
Let $S$ be a subset of $\mathbb{Z}^n$ and $f:\mathbb{Z}^k \to \mathbb{Z}$ be an operation.

\begin{itemize}
\item[(i)] 
If $S$ is $f$-closed, then
for each $1\le i<j \le n$, $\pi_{i,j}(S)$ is $f$-closed.

\item [(ii)]
If for each $1\le i<j \le n$, $S_{i,j}\be \bb{Z}^2$ is $f$-closed, then $\Join_{1\le i<j\le n} S_{i,j}$ is $f$-closed.
\end{itemize}

It follows that, if $S$ is $f$-closed, then 
$\Join_{1\le i<j\le n}\pi_{i,j}(S) $ is also $f$-closed.
\end{lemma}

\begin{proof}
\begin{itemize}
\item[$(i)$]
For any $(a^{(1)},b^{(1)}),\dots,(a^{(k)},b^{(k)})\in \pi_{i,j}(S)$, there exists $x^{(1)},\dots,x^{(k)}\in S$ such that $\proj_{i,j}(x^{(\ell)})=(a^{(\ell)},b^{(\ell)})$ for $1\le \ell\le k$.
Since $S$ is $f$-closed, 
$f(x^{(1)},\dots,x^{(k)})\in S$. 
Thus, $f((a^{(1)},b^{(1)}),\dots,(a^{(k)},b^{(k)}))=\proj_{i,j}(f(x^{(1)},\dots,x^{(k)}))\in \pi_{i,j}(S)$.

\item[$(ii)$] 
Take any $x^{(1)},\dots,x^{(k)}\in \Join_{1\le i< j\le n}S_{i,j}$. 
Then, for any $1\le \ell\le k$ and $1\le i<j\le n$, 
$\pi_{i,j}(x^{(\ell)})\in S_{i,j}$.
Since each $S_{i,j}$ is $f$-closed, 
$f(\pi_{i,j}(x^{(1)}),\dots,\pi_{i,j}(x^{(k)}))\in S_{i,j}$.
Since $\pi_{i,j}(f(x^{(1)},\dots,x^{(k)}))=f(\pi_{i,j}(x^{(1)}),\dots,\pi_{i,j}(x^{(k)}))$ for each $i$ and $j$, 
we have $f(x^{(1)},\dots,x^{(k)})\in \Join_{1\le i<j\le n}S_{i,j}$.
\end{itemize}   
\end{proof}

\begin{lemma}
\label{lemma:2-decompo-closure-operator}
Let $f:\mathbb{Z}^k \to \mathbb{Z}$ be an operation.
Then the operator that takes $S \be \mathbb{Z}^n$ as input  and outputs $\Join_{1\le i<j\le n}\cl_{f} (\pi_{i,j}(S)) (\be \mathbb{Z}^n)$ is a closure operator.
\end{lemma}

\begin{proof}
To show this, 
we check the following three conditions.
\begin{itemize}
\item[$(i)$]
Firstly, we show $S\be \Join_{1\le i<j\le n}\cl_{f} (\pi_{i,j}(S))$.
For any $x\in S$ and  for any $1\le i<j\le n$, 
it is clear that $\pi_{i,j}(x)\in \pi_{i,j}(S)\be \cl_{f} (\pi_{i,j}(S))$. 
Thus, $x\in \Join_{1\le i<j\le n}\cl_{f} (\pi_{i,j}(S))$.

\item[$(ii)$]
Secondly, we show that if $S_1\be S_2$ then 
$\Join_{1\le i<j\le n}\cl_{f} (\pi_{i,j}(S_1))\be \Join_{1\le i<j\le n}\cl_{f} (\pi_{i,j}(S_2))$.
Since $\Join_{1\le i<j\le n}, \cl_{f},$ and $\pi_{i,j}$ 
are monotonically increasing in terms of set inclusion, 
the composition 
$\Join_{1\le i<j\le n}\cl_{f} (\pi_{i,j}(\cdot))$ is also monotonically increasing.

\item[$(iii)$]
Finally, we show 
$\Join_{1\le i<j\le n}\cl_{f} (\pi_{i,j}(\Join_{1\le i<j\le n}\cl_{f} (\pi_{i,j}(S))))=\Join_{1\le i<j\le n}\cl_{f} (\pi_{i,j}(S))$.
By (ii), 
$\Join_{1\le i<j\le n}\cl_{f} (\pi_{i,j}(\Join_{1\le i<j\le n}\cl_{f} (\pi_{i,j}(S))))\supseteq\Join_{1\le i<j\le n}\cl_{f} (\pi_{i,j}(S))$. 
Thus, we show the inverse inclusion, 
that is 
\[\Join_{1\le i<j\le n}\cl_{f} (\pi_{i,j}(\Join_{1\le i<j\le n}\cl_{f} (\pi_{i,j}(S))))\be\Join_{1\le i<j\le n}\cl_{f} (\pi_{i,j}(S)).\]
It suffices to show that for each $1\le i'<j'\le n$ 
\[\cl_{f} (\pi_{i',j'}(\Join_{1\le i<j\le n}\cl_{f} (\pi_{i,j}(S))))\be\cl_{f} (\pi_{i',j'}(S)).\]
By minimality of $\cl_f$, 
it suffices to show that
\[ \pi_{i',j'}(\Join_{1\le i<j\le n}\cl_{f} (\pi_{i,j}(S)))\be\cl_{f} (\pi_{i',j'}(S)).\]
Now, $\pi_{i',j'}(\Join_{1\le i<j\le n}\cl_{f} (\pi_{i,j}(S)))\be\pi_{i',j'}(\pi_{i',j'}^{-1}(\cl_f(\pi_{i',j'}(S)))) = \cl_{f} (\pi_{i',j'}(S))$.   
\end{itemize}
\end{proof}

\begin{theorem}[2-decomposability and general closure property]
\label{thm:2-decompo-general-closure}
Let $S$ be a subset of $\mathbb{Z}^n$.
Let $\bf{P}$ be a property for subsets of $\mathbb{Z}^k$ which has closure operator ${\rm cl}_{\bf P}$ (for each $k=1,2,\dots$).
Consider the following properties.
\begin{itemize}
\item[(i)] $S$ is closed under some majority operation and closed in terms of $\bf{P}$, i.e., $S={\rm cl}_{\bf P}(S)$.
\item[(ii)] $S$ is closed under some majority operation and for each $1\le i<j \le n$, $\pi_{i,j}(S)$ is closed in terms of $\bf{P}$, i.e., $\pi_{i,j}(S)={\rm cl}_{\bf P}(\pi_{i,j}(S))$.
\item[(iii)] $S$ is 2-decomposable, i.e., $S = \Join_{1\le i<j\le n} \pi_{i,j}(S)$, and for each $1\le i<j \le n$, $\pi_{i,j}(S)$ is closed in terms of $\bf{P}$, i.e., $\pi_{i,j}(S)={\rm cl}_{\bf P}(\pi_{i,j}(S))$.
\item[(iv)] $S = \Join_{1\le i<j\le n}{\rm cl}_{\bf P} (\pi_{i,j}(S))$.
\end{itemize}
Then we have a chain of implications: 
(ii) $\nr$ (iii) $\nr$ (iv).
Moreover, (i) is incomparable to any other properties in general.
\end{theorem}

\begin{proof}
(ii) $\nr$ (iii) follows by \cref{thm:closed-near-unanimity<=>decompo}.
For (iii) $\nr$ (iv), 
since $\pi_{i,j}(S)={\rm cl}_{\bf P}(\pi_{i,j}(S))$ for each $1\le i<j \le n$, 
we have $S = \Join_{1\le i<j\le n} \pi_{i,j}(S) = \Join_{1\le i<j\le n}{\rm cl}_{\bf P} (\pi_{i,j}(S))$.
\end{proof}

\begin{theorem}[2-decomposability and closure property induced by operations]
\label{thm:2-decompo-closure-g}
Let $S$ be a subset of $\mathbb{Z}^n$ and 
$g:\mathbb{Z}^k \to \mathbb{Z}$ be an operation.
Consider the following properties.
\begin{itemize}
\item[(i)] $S$ is closed under some majority operation and $g$.
\item[(ii)] $S$ is closed under some majority operation, and for each $1\le i<j \le n$, $\pi_{i,j}(S)$ is $g$-closed.
\item[(iii)] $S = \Join_{1\le i<j\le n} \pi_{i,j}(S)$, and for each $1\le i<j \le n$, $\pi_{i,j}(S)$ is $g$-closed.
\item[(iv)] $S = \Join_{1\le i<j\le n}\cl_{g} (\pi_{i,j}(S))$.
\end{itemize}
Then we have a chain of implications: 
(i) $\douti$ (ii) $\nr$ (iii) $\douti$ (iv).
\end{theorem}

\begin{proof}
(i) $\nr$ (ii) directly follows by \cref{lemma:proj-closed} (i). 
For (ii) $\nr$ (i), since $S$ is closed under some majority operation, it is 2-decomposable, i.e., $S = \Join_{1\le i<j\le n}(\pi_{i,j}(S))$.
Then $S$ is $g$-closed from \cref{lemma:proj-closed} (ii).
(ii) $\nr$ (iii) and (iii) $\nr$ (iv) follow from \cref{thm:2-decompo-general-closure}.

Now, we show (iv) $\nr$ (iii).
In general, we have $S \subseteq \Join_{1\le i<j\le n} \pi_{i,j}(S)$ and $\Join_{1\le i<j\le n} \pi_{i,j}(S) \subseteq \Join_{1\le i<j\le n}{\rm cl}_{g} (\pi_{i,j}(S))$.
Hence, $S = \Join_{1\le i<j\le n}{\rm cl}_{g} (\pi_{i,j}(S))$ implies $S = \Join_{1\le i<j\le n} \pi_{i,j}(S)$.
From \cref{lemma:proj-closed} (ii), $S$ is $g$-closed.
Then from \cref{lemma:proj-closed} (i), for each $1\le i<j \le n$, $\pi_{i,j}(S)$ is $g$-closed.
This completes the proof.
\end{proof}

\begin{theorem}[2-decomposability and closure property induced by special operations]
\label{thm:2-decompo-closure-special-g}
Let $S$ be a subset of $\mathbb{Z}^n$ and 
$g:\mathbb{Z}^k \to \mathbb{Z}$ be an operation.
Assume that there exists a majority operation $h$ such that for any $R \subseteq \mathbb{Z}^2$ if $R$ is $g$-closed then it is $h$-closed.
Then the following are equivalent.
\begin{itemize}
\item[(i)] $S$ is closed under some majority operation  and $g$.
\item[(ii)] $S$ is closed under some majority operation, and for each $1\le i<j \le n$, $\pi_{i,j}(S)$ is $g$-closed.
\item[(iii)] $S = \Join_{1\le i<j\le n} \pi_{i,j}(S)$, and for each $1\le i<j \le n$, $\pi_{i,j}(S)$ is $g$-closed.
\item[(iv)] $S = \Join_{1\le i<j\le n}\cl_{g} (\pi_{i,j}(S))$.
\end{itemize}
\end{theorem}

\begin{proof}
By \cref{thm:2-decompo-closure-g}, it suffices to show that (iv) $\nr$ (i).
Assume $S = \Join_{1\le i<j\le n}\cl_{g} (\pi_{i,j}(S))$.
For each $1\le i<j \le n$, $\cl_{g} (\pi_{i,j}(S))$ is closed under majority operation $h$ by assumption.
Hence, $S$ is $h$-closed and $g$-closed by \cref{lemma:proj-closed}, implying that (i) holds.
\end{proof}

\begin{remark}
Contrast to \cref{thm:2-decompo-closure-special-g}, in \cref{thm:2-decompo-closure-g} (iii) does not imply (ii) for general operation $g$.
For example, when $g$ is a projection, (iii) means $S$ is 2-decomposable and (ii) means $S$ is closed under some majority operation, which are not equivalent 
(see \cref{sec:2-decomp-characterization}).
\end{remark}

\begin{theorem}[2-decomposability and closure property induced by $\cl_{\overline{\conv}}(\cdot) \cap \mathbb{Z}^n$]
\label{thm:2-decompo-closure-cl(conv)}
Let $S$ be a subset of $\mathbb{Z}^n$.
Consider the following properties.
\begin{itemize}
\item[(i)] $S$ is closed under some majority operation and closed in terms of $\cl_{\overline{\conv}}(\cdot) \cap \mathbb{Z}^n$, i.e., $S = \cl_{\overline{\conv}}(S) \cap \mathbb{Z}^n$.
\item[(ii)] $S$ is closed under some majority operation, and for each $1\le i<j \le n$, $\pi_{i,j}(S)$ is closed in terms of $\cl_{\overline{\conv}}(\cdot) \cap \mathbb{Z}^2$, i.e., $\pi_{i,j}(S) = \cl_{\overline{\conv}}(\pi_{i,j}(S)) \cap \mathbb{Z}^2$.
\item[(iii)] $S = \Join_{1\le i<j\le n} \pi_{i,j}(S)$, and for each $1\le i<j \le n$, $\pi_{i,j}(S)$ is closed in terms of $\cl_{\overline{\conv}}(\cdot) \cap \mathbb{Z}^2$, i.e., $\pi_{i,j}(S) = \cl_{\overline{\conv}}(\pi_{i,j}(S)) \cap \mathbb{Z}^2$.
\item[(iv)] $S = \Join_{1\le i<j\le n}(\cl_{\overline{\conv}}(\pi_{i,j}(S))\cap \mathbb{Z}^2)$.
\end{itemize}
Then we have a chain of implications: 
(ii) $\douti$ (iii) $\nr$ (iv) $\nr$ (i).
\end{theorem}

\begin{proof}
By \cref{thm:2-decompo-general-closure}, it suffices to show that (iii) $\nr$ (ii) and (iv) $\nr$ (i).

For (iii) $\nr$ (ii), 
assume that $S = \Join_{1\le i<j\le n} \pi_{i,j}(S)$, and for each $1\le i<j \le n$, $\pi_{i,j}(S) = \cl_{\overline{\conv}}(\pi_{i,j}(S)) \cap \mathbb{Z}^2$.
Then, each $\pi_{i,j}(S)$ is represented as a TVPI system with possibly infinitely many inequalities, and thus it is closed under the median operation (see \cref{rem:infiniteTVPI=>median-closed}).
It follows from \cref{lemma:proj-closed} (ii) that $S$ is closed the median operation.
Hence (ii) holds.

We then show that (iv) $\nr$ (i). 
(iv) means that $S$ is represented as a TVPI system with possibly infinitely many inequalities.
Hence, $S$ is closed under the median operation (see \cref{rem:infiniteTVPI=>median-closed}). Moreover, since $S$ is the set of integer vectors in an intersection of closed convex sets, we have $S = \cl_{\overline{\conv}}(S) \cap \mathbb{Z}^n$.
\end{proof}

\begin{proposition}[Hierarchy of closedness in 2-dimension]
\label{prop:hierarchy-in-2-dim}
For a subset $S \subseteq \mathbb{Z}^2$, 
consider the following properties.
\begin{itemize}
\item[(i)] $S$ is midpoint-neighbor-closed.
\item[(ii)] $S$ is closed under $g(x,y) = \lceil \frac{x+y}{2} \rceil$ and $h(x,y) = \lfloor \frac{x+y}{2} \rfloor$
\item[(iii)] $S$ is $\avgd$-closed.
\item[(iii)'] $S$ is integrally convex.
\item[(iv)] $S = \cl_{\overline{\conv}}(S)\cap \mathbb{Z}^2$.
\item[(v)] $S$ is closed under the median operation.
\item[(vi)] $S$ is closed under some majority operation.
\end{itemize}
Then we have a chain of implications: 
(i) $\nr$ (ii) $\nr$ (iii) $\douti$ (iii)' $\nr$ (iv) $\nr$ (v) $\nr$ (vi).
\end{proposition}

\begin{proof}
(i) $\nr$ (ii) is clear.
(ii) $\nr$ (iii) is shown in shown in \cite[Corollary 2]{TaT21}. 
(iii) $\nr$ (iii)' is also shown in \cite{TaT21} (see also \cref{lem:avgd=>int-conv}).

For (iii)' $\nr$ (iii), from \cref{prop:two-dim-IntConv=UTVPI}, 
if $S \subseteq \mathbb{Z}^2$ is integrally convex, then it can be represented by a UTVPI system. Then, $S$ is $\avgd$-closed from \cref{prop:UTVPI-median-avg-closed} in \cref{sec:pohyhedra-closure}.

For (iii) $\nr$ (iv), 
since $S$ is $\mu$-closed, it can be represented by a UTVPI system by \cref{lem:two-dim}.
Since a polyhedron is convex and topologically closed, (iv) holds.
(iv) $\nr$ (v) follows from \cref{rem:infiniteTVPI=>median-closed}. 
(v) $\nr$ (vi) is clear since the median operation is a majority operation.    
\end{proof}

\section{Polyhedral representation and closedness under operations}
\label{sec:pohyhedra-closure}

We here analyze the relationship between the set of integer vectors in 2-decomposable polyhedra and closedness under operations using the results in \cref{sec:2-decomp-closure}.

\subsection{Main results}
\label{subsec:main-results}

Using 
\cref{thm:2-decompo-closure-special-g,thm:2-decompo-closure-cl(conv),prop:hierarchy-in-2-dim}, 
we obtain the following.
The proof of (iii) is based on the results in \cref{subsec:remaining-proof}, which constitutes the most technical part in this paper.

\begin{theorem}\label{thm:characterization-by-closure-property}
Let $S$ be a subset of $\mathbb{Z}^n$.
\begin{itemize}
\item[(i)] $S$ is representable by a SVPI system if and only if 
$S$ is midpoint-neighbor-closed.
\item[(ii)] $S$ is representable by a DC system if and only if 
$S$ is closed under $g(x,y) = \lceil \frac{x+y}{2} \rceil$ and $h(x,y) = \lfloor \frac{x+y}{2} \rfloor$, and from \cref{thm:2-decompo-closure-special-g}, these are equivalent to, say,  
$S = \Join_{1\le i<j\le n}\cl_{g,h} (\pi_{i,j}(S))$.
\item[(iii)] $S$ is representable by a UTVPI system if and only if $S$ is closed under some majority operation and $\avgd$, and from \cref{thm:2-decompo-closure-special-g}, these are equivalent to, say, 
$S = \Join_{1\le i<j\le n}\cl_{\avgd} (\pi_{i,j}(S))$.
\item[(iv)] $S$ is representable by a TVPI system with possibly infinitely many inequalities if and only if 
$S = \Join_{1\le i<j\le n}\cl_{\overline{\conv}}(\pi_{i,j}(S))\cap \mathbb{Z}^n$.
\item[(v)] $S$ is closed under the median operation if and only if 
$S = \Join_{1\le i<j\le n}\cl_{\median} (\pi_{i,j}(S))$.
\item[(vi)] $S$ is closed under some majority operation if and only if 
$S = \Join_{1\le i<j\le n}\cl_{g} (\pi_{i,j}(S))$ for some majority operation $g:\mathbb{Z}^3 \to \mathbb{Z}$.
\end{itemize}
\end{theorem}

\begin{proof}
(i): 
We first show the only-if part.
Assume that $S$ is representable by a SVPI system.
Take any $x,y \in S$.
Since $S$ is representable by a SVPI system, interval $[x \wedge y, x \vee y ]_{\mathbb{Z}}$ is contained in $S$, i.e., $[x \wedge y, x \vee y ]_{\mathbb{Z}} \subseteq S$.
Since for each $z \in N(\frac{x+y}{2})$, $z \in [x \wedge y, x \vee y ]_{\mathbb{Z}}$ holds, 
$S$ is midpoint-neighbor-closed.

We then show the if part.
Assume that $S$ is midpoint-neighbor-closed. Then $S$ (and each $\pi_{i,j}(S)$) is also  closed under $g(x,y) = \lceil \frac{x+y}{2} \rceil$ and $h(x,y) = \lfloor \frac{x+y}{2} \rfloor$.
Thus, by (ii) shown below, $S$ (and each $\pi_{i,j}(S)$) is representable by a DC system.
It follows that $S$ is 2-decomposable, i.e., $S = \Join_{1\le i<j\le n}\pi_{i,j}(S)$, by \cref{prop:TVPI=>median-closed,thm:closed-near-unanimity<=>decompo}.
Then it suffices to show that each $\pi_{i,j}(S)$ is representable by a SVPI system.
Fix a DC representation of $\pi_{i,j}(S)$ and denote it by $Ax \ge b$.
Let $x_i - x_j \ge c$ be a difference inequality in $Ax \ge b$.
We show that $x_i - x_j = c$ intersects $\pi_{i,j}(S)$ at less than two points, which implies that removing $x_i - x_j \ge c$ from $Ax \ge b$ and appropriately adding $x_i\ge c_1$ and $-x_j \ge c_2$ to it result in a representation of $\pi_{i,j}(S)$ with less difference inequalities than $Ax \ge b$.
Assume otherwise that $x_i - x_j = c$ intersects $\pi_{i,j}(S)$ at two (distinct) points, say, $p$ and $q$.
We can assume that $(p_i.p_j) = (q_i+1,q_j+1)$.
Since $\pi_{i,j}(S)$ is midpoint-neighbor-closed, each $z \in N(\frac{p+q}{2})$ is in $\pi_{i,j}(S)$.
Then $\lfloor \frac{p_i+q_i}{2} \rfloor - \lceil \frac{p_j+q_j}{2} \rceil < \frac{p_i+q_i}{2} - \frac{p_j+q_j}{2} = c$.
Thus $z := (\lfloor \frac{p_i+q_i}{2} \rfloor, \lceil \frac{p_j+q_j}{2} \rceil)$ violates inequality $x_i - x_j \ge c$.
However, $z$ is in $N(\frac{p+q}{2})$, a contradiction.
Similarly, we can remove $x_j - x_i \ge c$ from $Ax \ge b$.
Hence, $\pi_{i,j}(S)$ is representable by a SVPI system.

(ii) is a well-known result in discrete convex analysis; see, e.g., Equation (5.20) in~\cite{Murota03}.

(iii): 
We first show the only-if part.
If $S$ is representable by a UTVPI system, 
by Theorem \ref{thm:main} (or Proposition \ref{prop:UTVPI-median-avg-closed}), 
$S$ is closed under $\avgd$ and $\median$.
Moreover, the $\median$ operation is a majority operation. 

We then show the if part.
Assume that $S$ is closed under some majority operation and $\avgd$.
Then $S$ is 2-decomposable, i.e., $S= \Join_{1\le i<j\le n}\pi_{i,j}(S)$ by \cref{thm:closed-near-unanimity<=>decompo}.
Since each $\pi_{i,j}(S)$ is closed under $\avgd$ by \cref{lemma:proj-closed}, each $\pi_{i,j}(S)$ can be represented by a UTVPI system from \cref{lem:two-dim}.
Hence, $S$ is representable by a UTVPI system.

(iv): 
We first show the only-if part.
Assume that $S$ is representable by a TVPI system with possibly infinitely many inequalities. 
This implies that $S = \{ x \in \mathbb{Z}^n \mid a_{k1}^{(i,j)}x_i + a_{k2}^{(i,j)}x_j \ge b_k^{(i,j)} (1 \le i < j \le n, k \in I_{i,j}) \}$ for some $a_{k1}^{(i,j)},a_{k2}^{(i,j)},b_k^{(i,j)} \in \mathbb{R}$ and a (possibly infinite) set $I_{i,j}$ of indices.
For each $1 \le i < j \le n$, let $S_{ij} = \{ x \in \mathbb{Z}^2 \mid a_{k1}^{(i,j)}x_i + a_{k2}^{(i,j)}x_j \ge b_k^{(i,j)} (k \in I_{i,j}) \}$.
Then, for $x \in \mathbb{Z}^n$, 
$x \in S$ if and only if $(x_i,x_j) \in S_{ij}$ for $1 \le i < j \le n$.
This implies that $S = \Join_{1\le i<j\le n}S_{ij}$.
Now, $S_{ij} = \cl_{\overline{\conv}}(S_{ij})\cap \mathbb{Z}^2$.
Hence, $S = \Join_{1\le i<j\le n}\cl_{\overline{\conv}}(\pi_{i,j}(S))\cap \mathbb{Z}^n$ holds.

We then show the if part.
Note that $S = \Join_{1\le i<j\le n}\cl_{\overline{\conv}}(\pi_{i,j}(S))\cap \mathbb{Z}^n$ is equivalent to the following condition:
\begin{itemize}
\item for $x \in \mathbb{Z}^n$, 
$x \in S$ if and only if $(x_i,x_j) \in \cl_{\overline{\conv}}(\pi_{i,j}(S))\cap \mathbb{Z}^2$ for $1 \le i < j \le n$.
\end{itemize}
Each $\cl_{\overline{\conv}}(\pi_{i,j}(S))\cap \mathbb{Z}^2$ can be represented as $\cl_{\overline{\conv}}(\pi_{i,j}(S))\cap \mathbb{Z}^2 = \{ x \in \mathbb{Z}^2 \mid a_{k1}^{(i,j)}x_i + a_{k2}^{(i,j)}x_j \ge b_k^{(i,j)} (k \in I_{i,j}) \}$ for some $a_{k1}^{(i,j)},a_{k2}^{(i,j)},b_k^{(i,j)} \in \mathbb{R}$ and a (possibly infinite) set $I_{i,j}$ of indices, since $\cl_{\overline{\conv}}(\pi_{i,j}(S))$ is the intersection of all the closed half-spaces containing $\proj_{i,j}(S)$ (see \cite[Corollary 11.5.1]{Roc70}).
Hence, for $x \in \mathbb{Z}^n$, 
$x \in S$ if and only if $a_{k1}^{(i,j)}x_i + a_{k2}^{(i,j)}x_j \ge b_k^{(i,j)} (1 \le i < j \le n, k \in I_{i,j})$.
This implies that $S = \{ x \in \mathbb{Z}^n \mid a_{k1}^{(i,j)}x_i + a_{k2}^{(i,j)}x_j \ge b_k^{(i,j)} (1 \le i < j \le n, k \in I_{i,j}) \}$.
Hence, $S$ is representable by a TVPI system with possibly infinitely many inequalities. 

(v) follows from \cref{thm:2-decompo-closure-special-g,prop:hierarchy-in-2-dim}.

(vi) follows from \cref{thm:2-decompo-closure-special-g}.
\end{proof}

Combining \cref{thm:characterization-by-closure-property} and the following proposition, we can characterize some closure operators.

\begin{proposition}
\label{prop:equiv-of-closure}
Let $\bf{Q}$ be a property for a subset of ${\mathbb{Z}}^n$ which has closure operator ${\rm cl}_{\bf Q}$, and 
$\bf{P}$ be a property for a subset of $\mathbb{Z}^2$ which has closure operator ${\rm cl}_{\bf P}$.
Then the following are equivalent. 
\begin{itemize}
\item[(i)] 
For any $S\subseteq \mathbb{Z}^n$, 
${\bf{Q}}(S)\iff S
=\Join_{1\le i<j\le n}{\rm cl}_{\bf P}(\pi_{i,j}(S))$
\item[(ii)]
For any $S\subseteq \mathbb{Z}^n$, 
${\rm cl}_{\bf{Q}}(S)=\Join_{1\le i<j\le n}{\rm cl}_{\bf P}(\pi_{i,j}(S))$.
\end{itemize}
\end{proposition}

\begin{proof}
We first show (i) $\nr$ (ii).
By (a generalization of) Lemma \ref{lemma:2-decompo-closure-operator}, 
$\cl_{\bf Q}$ and $\Join_{1\le i<j\le n}\cl_{\bf P}(\pi_{i,j}(\cdot))$ are both closure operators. 
Then, (i) means that the families of the closed set defined by these closure operators are the same. 
Thus, these closure operators are the same and (ii) holds.
We then show (ii) $\nr$ (i).
Assume that (ii) is true. 
Then, 
${\bf Q}(S)\douti S=\clq{S}{\douti}S=\Join_{1\le i<j\le n}\cl_{\bf P}(\pi_{i,j}(S))$.
\end{proof}

From \cref{thm:characterization-by-closure-property,prop:equiv-of-closure}, we see, for example, that $\cl_{\median,\avgd}$ is the same as $\Join_{1\le i<j\le n}\cl_{\avgd}(\proj_{i,j}(\cdot))$ (see also \cref{thm:main} below).
 
We mention the following.
\begin{proposition}
\label{prop:integrally-convex-2-decompo<=>UTVPI}
Let $S$ be a subset of $\mathbb{Z}^n$.
$S$ is representable by a UTVPI system if and only if it is integrally convex and 2-decomposable.
\end{proposition}

\begin{proof}
We first show the only-if part. 
From \cref{thm:characterization-by-closure-property}, $S$ is closed under some majority operation and $\avgd$.
Hence, $S$ is 2-decomposable and integrally convex by \cref{thm:closed-near-unanimity<=>decompo,lem:avgd=>int-conv}, respectively.

We then show the if part.
Since $S$ is 2-decomposable, $S = \Join_{1\le i<j\le n} \pi_{i,j}(S)$.
Moreover, it is known that the projection of an integrally convex set is integrally convex \cite[Theorem 3.1]{MoM19}.
Hence, each $\pi_{i,j}(S)$ is integrally convex.
Then each $\pi_{i,j}(S)$ can be represented by a UTVPI system (see \cref{prop:two-dim-IntConv=UTVPI}).
Hence, $S$ is representable by a UTVPI system.
\end{proof}

\begin{remark}
\label{rmk:two-dim-01-characterization}
When $S$ is a subset of $\bb{Z}^2$, the following are equivalent: (i) $S$ is representable by a UTVPI system, (ii) $S$ is $\avgd$-closed, and (iii) $S$ is integrally convex.
The equivalence of (i) and (iii) follows from \cref{prop:two-dim-IntConv=UTVPI}, and 
that of (ii) and (iii) follows from \cref{prop:hierarchy-in-2-dim}.

When $S$ is a subset of $\{0,1\}^n$, the following are equivalent: 
(i) $S$ is representable by a UTVPI system, (ii) $S$ is representable by a TVPI system, (iii) $S$ is $\median$-closed, and (iv) $S$ is closed under some majority operation.
To see this, it suffices to show that (iv) implies (i).
Assume that $S$ is closed under some majority operation.
Notice that there exists the unique majority operation over $\{0,1\}$, and 
$S$ is representable by a 2-CNF from \cite{Sch78}.
Since 2-CNF can be formulated as a UTVPI system, (i) holds.

We also remark that any subset $S$ of $\{0,1\}^n$ is $\avgd$-closed and integrally convex.
To see this, it suffices to show that any subset $S$ of $\{0,1\}^n$ is $\avgd$-closed.
Take any $S \be \{0,1\}^n$.
One can verify that for any $x,y \in \{0,1\}^n$, $\avgd(x,y)=x$.
Therefore, $S$ is $\avgd$-closed.
\end{remark}

\begin{remark}
We remark that property $S = \cl_{\overline{\conv}}(S)\cap \mathbb{Z}^n$
is not preserved by projection.
Namely, an analog of Lemma \ref{lemma:proj-closed} does not hold.
Consider the following example. 

Let $S\be \bb{Z}^3$ be a set defined as 
\[S=\{(0,0,0),(2,2,3)\}. \]
Then, in the figure below, 
$S$ is the black points and $\cl_{\overline{\conv}}(S)\be \bb{R}^3$ is the black line segment.  
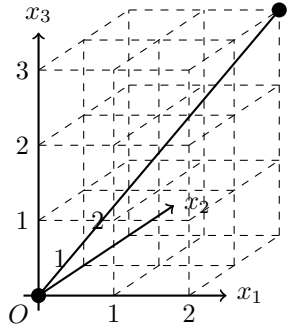
\begin{center}
\begin{tikzpicture}
\draw [->,thick] (-0.2,0) -- (2.5,0) node [right]{$x_1$};
\draw [->,thick] (0,-0.2) -- (0,3.5) node [above]{$x_3$};
\draw [->,thick] (0,0) -- (1.8,1.2) node [right]{$x_2$};

\fill[black](0,0)circle(0.1);    
\fill[black](3.2,3.8)circle(0.1);   
\draw [black, thick] (0,0) -- (3.2,3.8);

\draw [dashed] (1.2,0.8)--(3.2,0.8);   
\draw [dashed] (0.6,0.4)--(2.6,0.4);   
\draw [dashed] (2,0)--(3.2,0.8);       
\draw [dashed] (1,0)--(2.2,0.8);       

\draw [dashed] (1.2,1.8)--(3.2,1.8);   
\draw [dashed] (0.6,1.4)--(2.6,1.4);   
\draw [dashed] (0,1)--(2,1);           
\draw [dashed] (2,1)--(3.2,1.8);       
\draw [dashed] (1,1)--(2.2,1.8);       
\draw [dashed] (0,1)--(1.2,1.8);       

\draw [dashed] (1.2,2.8)--(3.2,2.8);   
\draw [dashed] (0.6,2.4)--(2.6,2.4);   
\draw [dashed] (0,2)--(2,2);           
\draw [dashed] (2,2)--(3.2,2.8);       
\draw [dashed] (1,2)--(2.2,2.8);       
\draw [dashed] (0,2)--(1.2,2.8);       

\draw [dashed] (1.2,3.8)--(3.2,3.8);   
\draw [dashed] (0.6,3.4)--(2.6,3.4);   
\draw [dashed] (0,3)--(2,3);           
\draw [dashed] (2,3)--(3.2,3.8);       
\draw [dashed] (1,3)--(2.2,3.8);       
\draw [dashed] (0,3)--(1.2,3.8);       

\draw [dashed] (0.6,0.4)--(0.6,3.4);    
\draw [dashed] (1.2,0.8)--(1.2,3.8);    

\draw [dashed] (1,0)--(1,3);           
\draw [dashed] (1.6,0.4)--(1.6,3.4);   
\draw [dashed] (2.2,0.8)--(2.2,3.8);   

\draw [dashed] (2,0)--(2,3);           
\draw [dashed] (2.6,0.4)--(2.6,3.4);   
\draw [dashed] (3.2,0.8)--(3.2,3.8);   

\draw (0,0)node[below left]{$O$};
\draw (1,0)node[below]{$1$}; 
\draw (2,0)node[below]{$2$}; 
\draw (0.5,0.5)node[left]{$1$}; 
\draw (1,1)node[left]{$2$}; 
\draw (0,2)node[left]{$2$}; 
\draw (0,1)node[left]{$1$}; 
\draw (0,3)node[left]{$3$}; 
\end{tikzpicture}
\end{center}
Thus, we have $S=\cl_{\overline{\conv}}(S)\cap \bb{Z}^3$.
However, the projection $\pi_{1,2}(S)$ does not satisfy this condition. 
In fact, since the projection $\pi_{1,2}(S)=\{(0,0), (2,2)\}\be \bb{Z}^2$, the convex closure 
$\cl_{\overline{\conv}}(\pi_{1,2}(S))=\{(x_1,x_2)\in \bb{R}^2~|~ 0\le x_1=x_2\le 2\}$. 
Thus, 
$\cl_{\overline{\conv}}(\pi_{1,2}(S))\cap \bb{Z}^2=\{(0,0),(1,1),(2,2)\}$. 
Therefore, 
we have $\pi_{1,2}(S) \neq \cl_{\overline{\conv}}(\pi_{1,2}(S))\cap \bb{Z}^2$.
\end{remark}

\begin{remark}
We also remark that, if a set $S \subseteq \mathbb{Z}^n$ is representable by a TVPI system, then it is $\median$-closed by \cref{prop:TVPI=>median-closed}.
However, even if $S$ is $\median$-closed and representable by a linear-inequality system (i.e., $S$ is the integer vectors in a polyhedron), 
it is not representable by a TVPI system in general as \cref{ex:median=>notTVPI} in \cref{subsec:median-closed-but-not-TVPI} shows.
\end{remark}

\subsection{Remaining proofs}
\label{subsec:remaining-proof}

The following lemma is shown in~\cite{TaT21} for a more general function variant, but a proof is given here for readability.

\begin{lemma}
\label{lem:avgd=>int-conv}
Let $S$ be a subset of $\mathbf{Z}^n$.
If $S$ is $\avgd$-closed, then it is integrally convex.
\end{lemma}
\begin{proof}
From \cref{thm:IntConvSet-characterization}, 
it suffices to show that $\frac{x+y}{2}\in \cl_{\conv}(S\cap N(\frac{x+y}{2}))$ for any $x,y\in S$.
Since $S$ is closed under $\avgd$, 
we have $\avgd(x,y), \avgd(y,x)\in S$. 
It is clear that  
$\avgd(x,y), \avgd(y,x)\in N(\frac{x+y}{2})$. 
Now, since $\avgd(x,y)+\avgd(y,x)=x+y$, 
we have 
$\frac{x+y}{2}=\frac{\avgd(x,y)+\avgd(y,x)}{2}\in \cl_{\conv}(\{\avgd(x,y), \avgd(y,x)\})\subseteq \cl_{\conv}(S\cap N(\frac{x+y}{2}))$. 
Hence, $S$ is integrally convex. 
\end{proof}

We have proved that any $S\subseteq \mathbb{Z}^n$ closed under $\avgd$ 
is integrally convex. 
The converse, however, does not hold in general as the following example shows.

\begin{example}
Let $S=\{ (-1,1,1),(0,0,1),(0,1,0),(1,0,0) \} \subseteq \{-1,0,1\}^3$.
The points in $S$ lies on the plane $x_1+x_2+x_3=1$ in three-dimension.
One can verify that 
\begin{align*}
\cl_{\conv}(S) = \cl_{\conv}(\{ (-1,1,1),(0,0,1),(0,1,0) \})\cup \cl_{\conv}(\{ (0,0,1),(0,1,0),(1,0,0) \})
\end{align*}
and $S$ is integrally convex.
On the other hand, as $\avgd((-1,1,1),(1,0,0)) = (0,1,1) \notin S$, 
S is not closed under $\avgd$.
Hence, integral convexity does not imply $\avgd$-closedness in general.
\end{example}

We now show the key result to show \cref{thm:characterization-by-closure-property} (iii).

\begin{theorem}
\label{thm:main}
A set $S \subseteq \mathbf{Z}^n$ is representable by a UTVPI system if and only if 
it is closed under $\avgd$ and $\median$.
\end{theorem}

\cref{thm:main} is divided into \cref{prop:UTVPI-median-avg-closed,prop:median-avg-closed->UTVPI}.

\begin{proposition}
\label{prop:UTVPI-median-avg-closed}
Let $S$ be a subset of $\mathbf{Z}^n$.
If $S$ is representable by a UTVPI system, then it is closed under $\avgd$ and $\median$.
\end{proposition}
\begin{proof}
Since a UTVPI polyhedron is a TVPI polyhedron, it follows that $S$ is closed under $\median$ by \cref{prop:TVPI=>median-closed}.
Thus we show that $S$ is closed under $\avgd$ in what follows.

Assume that $S$ is represented by a UTVPI system $Ax \ge b$ for some matrix $A$ and vector $b$, i.e., $S = \{x \in \bb{Z}^n \mid Ax \ge b\}$.
Without loss of generality, we assume that $b$ is an integer vector.
Let $p,q \in S$.
We show that $\avgd(p,q) \in S$ by showing that $A\avgd(p,q) \ge b$ holds.
For each $i = 1, \dots, m$, let $A_i$ be the $i$th row of $A$.
We have to show that $A_i\avgd(p,q) \ge b_i$ for each $i = 1, \dots, m$.
Since $Ax \ge b$ is a UTVPI system, each $A_ix \ge b_i$ is one of the forms of 
$x_j + x_k \ge b_i$, 
$x_j - x_k \ge b_i$, 
$-x_j + x_k \ge b_i$, 
$-x_j - x_k \ge b_i$, 
$x_j \ge b_i$, or
$-x_j \ge b_i$ for some $j,k \in \{1, \dots, n\}$ with $j \neq k$.
We use the fact that $A_ip \ge b_i$ and $A_iq \ge b_i$ (since $p$ and $q$ are in $S$) in what follows.

We only show the cases where $A_ix = x_j + x_k$ and $A_ix = x_j - x_k$.
The other cases can be shown similarly and thus are omitted.

\paragraph*{Case 1: $A_ix = x_j + x_k$.}
We show that $\avgd(p_j,q_j) + \avgd(p_k,q_k) \ge b_i$.
Note that from $p_j+p_k \ge b_i$ and $q_j+q_k \ge b_i$ we have $p_j + p_k + q_j + q_k \ge 2b_i$, implying that $\frac{p_j+q_j}{2} + \frac{p_k+q_k}{2} \ge b_i$.

\subparagraph*{Case 1.1: $p_j+q_j$ and $p_k+q_k$ are even.}
In this case, we have $\avgd(p_j,q_j) = \frac{p_j+q_j}{2}$ and $\avgd(p_k,q_k) = \frac{p_k+q_k}{2}$.
Thus, we have 
\begin{align*}
\avgd(p_j,q_j) + \avgd(p_k,q_k) = \frac{p_j+q_j}{2} + \frac{p_k+q_k}{2} \ge b_i.
\end{align*}

\subparagraph*{Case 1.2: $p_j+q_j$ is even and $p_k+q_k$ is odd.}
In this case, $p_j + p_k + q_j + q_k \ge 2b_i$ implies that $p_j + p_k + q_j + q_k \ge 2b_i + 1$ since the left-hand side is odd.
Thus, we have $\frac{p_j+q_j}{2} + \frac{p_k+q_k}{2} - \frac{1}{2} \ge b_i$.
Since $\avgd(p_j,q_j) = \frac{p_j+q_j}{2}$ and $\avgd(p_k,q_k) \ge \frac{p_k+q_k}{2} - \frac{1}{2}$, 
we have 
\begin{align*}
\avgd(p_j,q_j) + \avgd(p_k,q_k) \ge \frac{p_j+q_j}{2} + \frac{p_k+q_k}{2} - \frac{1}{2} \ge b_i.
\end{align*}

\subparagraph*{Case 1.3: $p_j+q_j$ is odd and $p_k+q_k$ is even.}
We can show $\avgd(p_j,q_j) + \avgd(p_k,q_k) \ge b_i$ in a similar way as in Case 1.2.

\subparagraph*{Case 1.4: $p_j+q_j$ and $p_k+q_k$ are odd.}
We further divide into the cases of $p_j > q_j$, $p_k > q_k$, or ($p_j \le q_j$ and $p_k \le q_k$).

If $p_j > q_j$, then $\avgd(p_j,q_j) = \frac{p_j+q_j}{2} + \frac{1}{2}$ since $p_j+q_j$ is odd.
Moreover, $\avgd(p_k,q_k) \ge \frac{p_k+q_k}{2} - \frac{1}{2}$.
Hence, we have 
\begin{align*}
\avgd(p_j,q_j) + \avgd(p_k,q_k) \ge \frac{p_j+q_j}{2} + \frac{1}{2} + \frac{p_k+q_k}{2} - \frac{1}{2} \ge b_i.
\end{align*}

If $p_k > q_k$, then $\avgd(p_k,q_k) = \frac{p_k+q_k}{2} + \frac{1}{2}$ since $p_k+q_k$ is odd.
Moreover, $\avgd(p_j,q_j) \ge \frac{p_j+q_j}{2} - \frac{1}{2}$.
Hence, we have 
\begin{align*}
\avgd(p_j,q_j) + \avgd(p_k,q_k) \ge \frac{p_j+q_j}{2} - \frac{1}{2} + \frac{p_k+q_k}{2} + \frac{1}{2} \ge b_i.
\end{align*}

Finally, consider the case of $p_j \le q_j$ and $p_k \le q_k$.
We have $2p_j \le p_j + q_j$, and thus $p_j \le \frac{p_j+q_j}{2}$.
Since $p_j$ is an integer we have $p_j \le \lfloor \frac{p_j+q_j}{2} \rfloor = \avgd(p_j,q_j)$.
Similarly, we have $p_k \le \avgd(p_k,q_k)$
Therefore, we have 
\begin{align*}
\avgd(p_j,q_j) + \avgd(p_k,q_k) \ge p_j+p_k \ge b_i.
\end{align*}

\paragraph*{Case 2: $A_ix = x_j - x_k$.}
We show that $\avgd(p_j,q_j) - \avgd(p_k,q_k) \ge b_i$.
Note that from $p_j-p_k \ge b_i$ and $q_j-q_k \ge b_i$ we have $p_j - p_k + q_j - q_k \ge 2b_i$, implying that $\frac{p_j+q_j}{2} - \frac{p_k+q_k}{2} \ge b_i$.
Then the proof goes similarly to that of Case 1.

Now, we have shown that 
$\avgd(p.q)$ satisfies each inequality $A_ix \ge b_i$ ($i=1,\dots, m$).
Hence, $\avgd(p.q) \in S$ holds.
Therefore, $S$ is closed under $\avgd$.
This completes the proof.
\end{proof}

As a corollary of the 
\cref{prop:UTVPI-median-avg-closed},
we have the following result.

\begin{corollary}
\label{cor:avg-in-UTVPI}
For a set $S\subseteq \mathbb{Z}^n$, 
the directed discrete midpoint closure $\cl_{\avgd}(S)$ of $S$ is included in the UTVPI closure $\cl_{\utvpi}(S)$ of $S$.
\end{corollary}
\begin{proof}
Since $\cl_{\utvpi}(S)$ is representable by a UTVPI system, 
by \cref{prop:UTVPI-median-avg-closed},
$\cl_{\utvpi}(S)$ is closed under $\avgd$. 
Now, $\cl_{\avgd}(S)$ is the (inclusion-wise) smallest set that is closed under $\avgd$ and  includes $S$. 
Therefore, 
$\cl_{\avgd}(S)\subseteq \cl_{\utvpi}(S)$. 
\end{proof}

We note that there exists a TVPI polyhedron where its integer vectors are not $\avgd$-closed.

\begin{example}
Consider the following TVPI system.

\begin{align}
\left\{
\begin{array}{l}
x_1 + 2x_2 \ge  2\\
0 \le x_1,x_2 \le 2.
\end{array}
\right.
\end{align}
The set $S$ of integer vectors satisfying the system is 
\begin{align*}
S = \{(0,1),(0,2),(1,1),(1,2),(2,0),(2,1),(2,2)\}.
\end{align*}
Since $\avgd((2,0),(0,1))=(\avgd(2,0),\avgd(0,1)) = (1,0) \notin S$, 
$S$ is not $\avgd$-closed.
\end{example}

We now show the opposite direction of \cref{prop:UTVPI-median-avg-closed}.
The following is a key lemma to show our result.

\begin{lemma}\label{lem:two-dim}
Let $S$ be a subset of $\mathbf{Z}^2$.
If $S$ is $\avgd$-closed, then $S$ is representable by a UTVPI system.
\end{lemma}
\begin{proof}
See \cref{rmk:two-dim-01-characterization}.
\end{proof}

\begin{proposition}
\label{prop:median-avg-closed->UTVPI}
If a set $S \subseteq \mathbf{Z}^n$ is closed under $\avgd$ and $\median$, 
then it is representable by a UTVPI system.
\end{proposition}
\begin{proof}
First, since $S$ is $\median$-closed, it is 2-decomposable, i.e., $S = \Join_{1 \le i<j \le n}\pi_{i,j}(S)$.
Since each $\pi_{i,j}(S)$ is two-dimensional and $\avgd$-closed by \cref{lemma:proj-closed}, it is representable by a UTVPI system $A^{(ij)}(x_i,x_j)^T\ge b^{(ij)}$ by \cref{lem:two-dim}.
Then, clearly, $S = \{ x \in \mathbb{Z}^n \mid A^{(ij)}(x_i,x_j)^T\ge b^{(ij)} (\forall i,j) \}$).
Hence, $S$ is representable by a UTVPI system.
\end{proof}

\begin{proof}[Proof of \cref{thm:main}]
It follows from \cref{prop:UTVPI-median-avg-closed,prop:median-avg-closed->UTVPI}.
\end{proof}

\subsection{TVPI representability cannot be characterized by median-closedness}
\label{subsec:median-closed-but-not-TVPI}

One may expect that if 
the set of integer vectors in a polyhedron is $\median$-closed, then 
it is representable by a TVPI system (i.e., the opposite of \cref{prop:TVPI=>median-closed}).
However, this is not the case, as the following example shows.

\begin{example}\label{ex:median=>notTVPI}
Let 
\begin{align*}
A=
\begin{pmatrix}
2 & 0 & -1    \\
0 & 2 & -1   \\
-2 & -2 & 3 \\
0 & 0 & -1  \\
\end{pmatrix}\ and \ 
b=\begin{pmatrix}
0\\0\\0\\2
\end{pmatrix}.
\end{align*}
Consider a set $S\be \mathbb{Z}^3$  defined as
\[S:=\{x \in \mathbb{Z}^3\mid Ax \ge b\}\ (=\{(0,0,0),(1,1,2),(2,1,2), (1,2,2)\}).\]
Then $S$ is the black points in the figure below. 

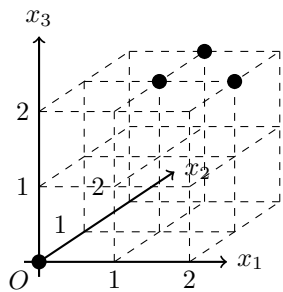
\begin{center}
\begin{tikzpicture}
\draw [->,thick] (-0.2,0) -- (2.5,0) node [right]{$x_1$};
\draw [->,thick] (0,-0.2) -- (0,3) node [above]{$x_3$};
\draw [->,thick] (0,0) -- (1.8,1.2) node [right]{$x_2$};

\fill[black](0,0)circle(0.1);      
\fill[black](1.6,2.4)circle(0.1);  
\fill[black](2.6,2.4)circle(0.1);  
\fill[black](2.2,2.8)circle(0.1);  

\draw [dashed] (1.2,0.8)--(3.2,0.8);   
\draw [dashed] (0.6,0.4)--(2.6,0.4);   
\draw [dashed] (2,0)--(3.2,0.8);       
\draw [dashed] (1,0)--(2.2,0.8);       

\draw [dashed] (1.2,1.8)--(3.2,1.8);   
\draw [dashed] (0.6,1.4)--(2.6,1.4);   
\draw [dashed] (0,1)--(2,1);           
\draw [dashed] (2,1)--(3.2,1.8);       
\draw [dashed] (1,1)--(2.2,1.8);       
\draw [dashed] (0,1)--(1.2,1.8);       

\draw [dashed] (1.2,2.8)--(3.2,2.8);   
\draw [dashed] (0.6,2.4)--(2.6,2.4);   
\draw [dashed] (0,2)--(2,2);           
\draw [dashed] (2,2)--(3.2,2.8);       
\draw [dashed] (1,2)--(2.2,2.8);       
\draw [dashed] (0,2)--(1.2,2.8);       

\draw [dashed] (0.6,0.4)--(0.6,2.4);    
\draw [dashed] (1.2,0.8)--(1.2,2.8);    

\draw [dashed] (1,0)--(1,2);           
\draw [dashed] (1.6,0.4)--(1.6,2.4);   
\draw [dashed] (2.2,0.8)--(2.2,2.8);   

\draw [dashed] (2,0)--(2,2);           
\draw [dashed] (2.6,0.4)--(2.6,2.4);   
\draw [dashed] (3.2,0.8)--(3.2,2.8);   

\draw (0,0)node[below left]{$O$};
\draw (1,0)node[below]{$1$}; 
\draw (2,0)node[below]{$2$}; 
\draw (0.5,0.5)node[left]{$1$}; 
\draw (1,1)node[left]{$2$}; 
\draw (0,2)node[left]{$2$}; 
\draw (0,1)node[left]{$1$}; 
\end{tikzpicture}
\end{center}

We now show that $S$ is $\median$-closed, 
whereas it is not representable by a TVPI system.

First we see that $S$ is $\median$-closed. Indeed, it is trivial when at least two of three arguments of $\median$ are the same. 
Otherwise, we have
\begin{gather*}
\median((0,0,0),(1,1,2),(2,1,2))=(1,1,2)\in S,\\
\median((0,0,0),(1,1,2),(1,2,2))=(1,1,2)\in S, \\
\median((0,0,0),(2,1,2),(1,2,2))=(1,1,2)\in S, \\
\median((1,1,2),(2,1,2),(1,2,2))=(1,1,2)\in S. 
\end{gather*}
Thus, $S$ is $\median$-closed.

We next show that $S$ is not representable by a TVPI system. 
Assume otherwise that we have a TVPI representation 
$S=\{x\in \mathbb{Z}^3 \mid A'x\ge b'\}$.
Then the following equality holds. 
\begin{equation}\label{eq:ofS}
S=\proj_{1,2}^{-1}(\cl_{\overline{\conv}}({\proj_{1,2}(S)}))\cap \proj_{1,3}^{-1}(\cl_{\overline{\conv}}(\proj_{1,3}(S)))\cap \proj_{2,3}^{-1}(\cl_{\overline{\conv}}(\proj_{2,3}(S)))\cap \mathbb{Z}^3
\end{equation}

Indeed, since $S \be \proj_{i,j}^{-1}(\proj_{i,j}(S))\subseteq \proj_{i,j}^{-1}(\cl_{\overline{\conv}}(\proj_{i,j}(S)))$, the left-hand side of $(2)$ is included in the right-hand side. 
Conversely, letting $p=(p_1,p_2,p_3)$ be a vector in the right-hand side of $(2)$, we show that $p\in S$.  
It suffices to show that any two-variable inequality $a_1x_i+a_2x_j\ge b_1$ of $A'x\ge b'$ is satisfied by $p$, i.e., $a_1p_i+a_2p_j\ge b_1$. 
Since, $p\in \proj_{i,j}^{-1}(\cl_{\overline{\conv}}(\proj_{i,j}(S)))$, we have $\proj_{i,j}(p)\in \cl_{\overline{\conv}}(\proj_{i,j}(S))$. 
Moreover, since $\proj_{i,j}(S)\subseteq \{(x_i,x_j)\in \mathbb{R}^2 \mid a_1x_i+a_2x_j\ge b_1\}$, we have $\cl_{\overline{\conv}}(\proj_{i,j}(S))\subseteq \{(x_i,x_j)\in \mathbb{R}^2 \mid a_1x_i+a_2x_j\ge b_1\}$. 
Hence, we have $a_1p_i+a_2p_j\ge b_1$.

Now, $\cl_{\overline{\conv}}(\proj_{1,2}(S))$, $\cl_{\overline{\conv}}(\proj_{1,3}(S))$, and $\cl_{\overline{\conv}}\proj_{2,3}(S))$ are the following grayed out subsets
where the black circles are the projection of the points in $S$,
and the squares are integer points in the grayed out subset that are not in the projection of the points in $S$.

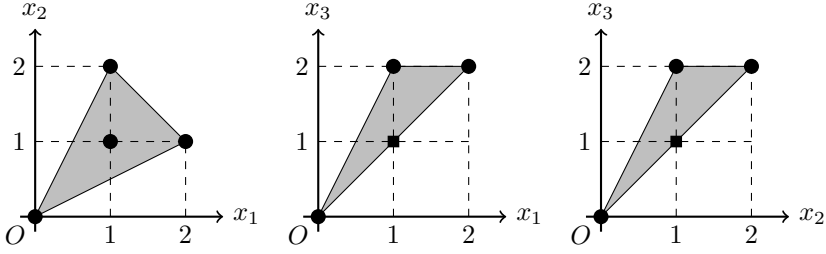
\begin{tikzpicture}
\coordinate(A)at(0,0);
\coordinate(B)at(2,1);
\coordinate(C)at(1,2);
\fill[lightgray](A)--(B)--(C)--cycle;

\draw [->,thick] (-0.2,0) -- (2.5,0) node [right]{$x_1$};
\draw [->,thick] (0,-0.2) -- (0,2.5) node [above]{$x_2$};
\node [below left] (0,0){$O$};
\draw [dashed] (0,1)--(2,1);
\draw (0,1)node[left]{$1$};
\draw [dashed] (0,2)--(1,2);
\draw (0,2)node[left]{$2$};
\draw [dashed] (1,0)--(1,2);
\draw (1,0)node[below]{$1$};
\draw [dashed] (2,0)--(2,1);
\draw (2,0)node[below]{$2$};
\fill[black](0,0)circle(0.1);
\fill[black](1,2)circle(0.1);
\fill[black](2,1)circle(0.1);
\fill[black](1,1)circle(0.1);
\draw (0,0)--(2,1);
\draw (0,0)--(1,2);
\draw (2,1)--(1,2);
\end{tikzpicture}
\begin{tikzpicture}
\coordinate(A)at(0,0);
\coordinate(B)at(1,2);
\coordinate(C)at(2,2);
\fill[lightgray](A)--(B)--(C)--cycle;

\draw [->,thick] (-0.2,0) -- (2.5,0) node [right]{$x_1$};
\draw [->,thick] (0,-0.2) -- (0,2.5) node [above]{$x_3$};
\node [below left] (0,0){$O$};
\draw [dashed] (0,1)--(2,1);
\draw (0,1)node[left]{$1$};
\draw [dashed] (0,2)--(2,2);
\draw (0,2)node[left]{$2$};
\draw [dashed] (1,0)--(1,2);
\draw (1,0)node[below]{$1$};
\draw [dashed] (2,0)--(2,2);
\draw (2,0)node[below]{$2$};
\fill[black](0,0)circle(0.1);
\fill[black](1,2)circle(0.1);
\fill[black](2,2)circle(0.1);
\draw [black] plot [only marks, mark=square*] coordinates {(1,1)};
\draw (0,0)--(1,2);
\draw (0,0)--(2,2);
\draw (1,2)--(2,2);

\end{tikzpicture}
\begin{tikzpicture}
\coordinate(A)at(0,0);
\coordinate(B)at(1,2);
\coordinate(C)at(2,2);
\fill[lightgray](A)--(B)--(C)--cycle;

\draw [->,thick] (-0.2,0) -- (2.5,0) node [right]{$x_2$};
\draw [->,thick] (0,-0.2) -- (0,2.5) node [above]{$x_3$};
\node [below left] (0,0){$O$};
\draw [dashed] (0,1)--(2,1);
\draw (0,1)node[left]{$1$};
\draw [dashed] (0,2)--(2,2);
\draw (0,2)node[left]{$2$};
\draw [dashed] (1,0)--(1,2);
\draw (1,0)node[below]{$1$};
\draw [dashed] (2,0)--(2,2);
\draw (2,0)node[below]{$2$};
\fill[black](0,0)circle(0.1);
\fill[black](1,2)circle(0.1);
\fill[black](2,2)circle(0.1);
\draw [black] plot [only marks, mark=square*] coordinates {(1,1)};
\draw (0,0)--(1,2);
\draw (0,0)--(2,2);
\draw (1,2)--(2,2);
\end{tikzpicture}

We see that $(1,1,1)$ is included in the right-hand side of $(\ref{eq:ofS})$.
But $(1,1,1)\notin S$ which means $S$ is not representable by a TVPI system.
\end{example}

\section{Characterizing 2-decomposability}
\label{sec:2-decomp-characterization}
We now turn our attention to characterizing 2-decomposability.

\begin{definition}[Partial operation]
An operation $D \to \bb{Z}$, where $D\be \bb{Z}^k$, is called a ($k$-ary) \emph{partial operation} (over $\bb{Z}$).
$D$ is called the domain of $f$ and denoted by $\dom(f)$.
If $x \in D$, then we say $f(x)$ is defined, and otherwise (i.e., if $x \notin D$) undefined.
\end{definition}

Since a partial operation is not defined for some inputs, we can define closedness under such operation in several ways.

\begin{definition}
Let $S$ be a subset of $\mathbb{Z}^n$.
Let $f:D\to \bb{Z}$, where $D\be \bb{Z}^k$, be a partial operation. 
\begin{itemize}
\item[(i)] We say $S$ is \emph{strongly closed under $f$} (or, \emph{strongly $f$-closed}) if for any $x^{(1)},\dots,x^{(k)}\in S$, there exists $x \in S$ that satisfies $f(x^{(1)}_i,\dots,x^{(k)}_i)=x_i$ for each $1\le i\le n$ such that $f(x^{(1)}_i,\dots,x^{(k)}_i)$ is defined.
\item[(ii)]
We say $S$ is \emph{weakly closed under $f$} (or, \emph{weakly $f$-closed}) if for any $x^{(1)},\dots,x^{(k)}\in S$ such that $f(x^{(1)}_i,\dots,x^{(k)}_i)$ is defined for all $1\le i\le n$, 
we have $f(x^{(1)},\dots,x^{(k)})\in S$.
\end{itemize}    
\end{definition}

\begin{definition}
Let $S$ be a subset of $\bb{Z}^n$.
Let $f_i:\bb{Z}^n \to \bb{Z}$ be an operation associated with the $i$th coordinate for each $1\le i\le n$. 
We say $S$ is \emph{closed under $(f_i)_{1\le i\le n}$} (or, \emph{$(f_i)_{1\le i\le n}$-closed}) if for any $x^{(1)},\dots,x^{(k)}\in S$
\begin{align*}
(f_1(x^{(1)}_1,\dots,x^{(k)}_1),\dots,f_n(x^{(1)}_n,\dots,x^{(k)}_n)) \in S.
\end{align*}
\end{definition}

\begin{definition}
The partial majority operation $\maj_{p}:\dom (\maj_{p}) \rightarrow \mathbb{Z}$ is defined as follows.
Firstly, its domain $\dom (\maj_{p}) = \{ (x,x,y) \mid x,y \in \mathbb{Z} \} \cup \{ (x,y,x) \mid x,y \in \mathbb{Z} \} \cup \{ (y,x,x) \mid x,y \in \mathbb{Z} \} \be \mathbb{Z}^3$.
Then, for all $x,y \in \mathbb{Z}$, 
\begin{align*}
\maj_{p}(x,x,y) = \maj_{p}(x,y,x) = \maj_{p}(y,x,x) = x.
\end{align*}
\end{definition}

We first summarize the relationships between 2-decomposability and closedness under (partial) operations.

\begin{proposition}
\label{prop:2-decompo-related-property}
For a subset $S \subseteq \mathbb{Z}^n$, 
consider the following properties.
\begin{itemize}
\item[(i)] $S$ is closed under some majority operation.
\item[(ii)] We can associate a majority operation $f_i$ for each coordinate $i$ of $S$ and $S$ is $(f_i)_{i}$-closed.
\item[(iii)] $S$ is strongly closed under the partial majority operation.
\item[(iv)] $S = \Join_{1\le i<j\le n} \pi_{i,j}(S)$ (i.e., $S$ is 2-decomposable).
\item[(v)] $S$ is weakly closed under the partial majority operation.
\end{itemize}
For each property, we also consider the hereditary version of it, denoted by $(\cdot)^h$, which requires that every projection of $S$ satisfies the property.
Then we have a chain of implications: 
(i) $\Rightarrow$ (ii) $\Rightarrow$ (iii) $\Rightarrow$ (iv) $\Rightarrow$ (v).
We also have the following equivalences: 
(i) $\Leftrightarrow$ $(i)^h$, 
(ii) $\Leftrightarrow$ $(ii)^h$, and 
(iii) $\Leftrightarrow$ $(iii)^h$ $\Leftrightarrow$ $(iv)^h$ $\Leftrightarrow$ $(v)^h$.
\end{proposition}

\begin{proof}
For (i) $\nr$ (ii), 
if $S$ is closed under majority operation $f$, 
we can take $f_i:=f$ for all $1\le i\le n$. 
For (ii) $\nr$ (iii), 
consider arbitrary $x^{(1)},x^{(2)},x^{(3)}\in S$.
We know that $(f_1(x^{(1)}_1,x^{(2)}_1,x^{(3)}_1),\dots,f_n(x^{(1)}_n,x^{(2)}_n,x^{(3)}_n)\in S$.
Moreover, for each $1\le i\le n$ such that $(x^{(1)}_i,x^{(2)}_i,x^{(3)}_i) \in \dom(\maj_p)$ we have $\maj_p(x^{(1)}_i,x^{(2)}_i,x^{(3)}_i) = f_i(x^{(1)}_i,x^{(2)}_i,x^{(3)}_i)$ since $f_i$ is a majority operation.
Hence, $S$ is strongly closed under $\maj_p$.

For (iii) $\nr$ (iv), 
it is clear that $S\be \Join_{1\le i<j\le n}\pi_{i,j}(S)$.
Thus, we show the inverse inclusion.
This is trivial when $n=2$.
Assume that $n \ge 3$ in the following.
Take $x\in \Join_{1\le i<j\le n}\pi_{i,j}(S)$.
Now, for $ 3\le \ell\le n$, 
we show for any subset $I\be [n]$ of cardinality $\ell$ there exists $y \in S$ such that $x_i=y_i$ for all $i \in I$.
We show this by induction on $\ell$.
First, consider the case of $\ell=3$.
Without of loss generality, 
we consider $I=\{1,2,3\}$.
Since $x \in \Join_{1\le i<j\le n}\pi_{i,j}(S)$, 
there exist $x^{(1,2)},x^{(1,3)},x^{(2,3)}\in S$ 
such that 
$\pi_{1,2}(x^{(1,2)})=\pi_{1,2}(x)$, 
$\pi_{1,3}(x^{(1,3)})=\pi_{1,3}(x)$, and 
$\pi_{2,3}(x^{(2,3)})=\pi_{2,3}(x)$.
Therefore, for each $i \in I$, $\maj_p(x^{(1,2)}_i,x^{(1,3)}_i,x^{(2,3)}_i)$ is defined and equal to $x_i$.
Since $S$ is strongly closed under $\maj_p$, 
there exists $y\in S$ 
such that $y_i = \maj_p(x^{(1,2)}_i,x^{(1,3)}_i,x^{(2,3)}_i)$ for all $i \in I$. 
This shows the case of $\ell=3$.
For $\ell > 3$, Without of loss generality, 
we consider $I=\{1,\dots,\ell\}$. 
By induction hypothesis for 
$I_{1,2}=\{1,2,4,\dots,\ell\},
I_{1,3}=\{1,3,4,\dots,\ell\}$, 
and $I_{2,3}=\{2,3,4,\dots,\ell\}$, 
we can take $x^{(1,2)},x^{(1,3)},x^{(2,3)} \in S$ such that $x^{(j,k)}_i=x_i$ for all $i \in I_{j,k}$ where $(j,k) =(1,2), (1,3), (2,3)$.
It follows that for all $i \in I$ we have $\maj_p(x^{(1,2)}_i,x^{(1,3)}_i,x^{(2,3)}_i) = x_i$.
Since $S$ is strongly closed under $\maj_p$, 
there exists $y\in S$ 
such that $y_i = \maj_p(x^{(1,2)}_i,x^{(1,3)}_i,x^{(2,3)}_i)$ for all $i \in I$. 
This shows the case of general $\ell$.
Therefore, $S=\Join_{1\le i<j\le n}\pi_{i,j}(S)$.

We then show (iv) $\nr$ (v).
We assume $S=\Join_{1\le i<j\le n}\pi_{i,j}(S)$.
For any $x^{(1)},x^{(2)},x^{(3)} \in S$ such that $(x^{(1)}_i,x^{(2)}_i,x^{(3)}_i) \in \dom(\maj_p)$ for all $1\le i\le n$, let $y:= \maj_p(x^{(1)},x^{(2)},x^{(3)})$.
By the definition of $y$, for each $1\le i<j\le n$, $\pi_{i,j}(y)$ is equal to at least one of $\pi_{i,j}(x^{(1)}), \pi_{i,j}(x^{(2)}),$ and $\pi_{i,j}(x^{(3)})$. 
Thus, $\pi_{i,j}(y) \in \pi_{i,j}(S)$. 
Since $S=\Join_{1\le i<j\le n}\pi_{i,j}(S)$, 
we have $y\in S$.
Hence, $S$ is weakly closed under $\maj_p$.

For the relationships with the hereditary versions, 
it is clear that $(\cdot)^h$ implies $(\cdot)$.
Hence, it suffices to show that (i) $\nr$ (i$)^h$, (ii) $\nr$ (ii$)^h$, (iii) $\nr$ (iii$)^h$, and (v$)^h$ $\nr$ (iii$)^h$.

For (i) $\nr$ (i$)^h$, we can show that closedness under operations is preserved by projection in a similar way as \cref{lemma:proj-closed}.

For (ii) $\nr$ (ii$)^h$, let $I \be \{1,\dots, n \}$ and take $x^{(1)},x^{(2)},x^{(3)} \in \proj_I(S)$.
Then for each $j=1,2,3$ there exists $y^{(j)} \in S$ such that $\proj_I(y^{(j)}) = x^{(j)}$.
By assumption, $y:=(f_1(y^{(1)}_1,y^{(2)}_1,y^{(3)}_1),\dots,f_n(y^{(1)}_n,y^{(2)}_n,y^{(3)}_n)) \in S$.
Therefore, $(f_i(x^{(1)}_i,x^{(2)}_i,x^{(3)}_i))_{i \in I} = \proj_I(y) \in \proj_I(S)$.

For (iii) $\nr$ (iii$)^h$, let $I \be \{\ 1,\dots, n \}$ and take $x^{(1)},x^{(2)},x^{(3)} \in \proj_I(S)$.
Then for each $j=1,2,3$ there exists $y^{(j)} \in S$ such that $\proj_I(y^{(j)}) = x^{(j)}$.
By assumption, there exists $y \in S$ that satisfies $y_i = \maj_p(y^{(1)}_i,y^{(2)}_i,y^{(3)}_i)$ for all $1 \le i \le n$ such that $\maj_p(y^{(1)}_i,y^{(2)}_i,y^{(3)}_i)$ is defined.
Then $\proj_I(y)$ satisfies that $y_i = \maj_p(x^{(1)}_i,x^{(2)}_i,x^{(3)}_i)$ for all $i \in I$ such that $\maj_p(x^{(1)}_i,x^{(2)}_i,x^{(3)}_i)$ is defined.

Finally, for (v$)^h$ $\nr$ (iii$)^h$, it suffices to show that (v$)^h$ $\nr$ (iii) (since we have shown (iii) $\nr$ (iii$)^h$).
Take $x^{(1)},x^{(2)},x^{(3)} \in S$.
Let $I = \{ i \in \{1,\dots,n\} \mid \maj_p(x^{(1)}_i,x^{(2)}_i,x^{(3)}_i)\text{ is defined} \}$.
By assumption, $\proj_I(S)$ is weakly $\maj_p$-closed.
Therefore, there exists $y \in \proj_I(S)$ such that $y_i = \maj_p(x^{(1)}_i,x^{(2)}_i,x^{(3)}_i)$ for all $i \in I$.
Then there exists $x \in S$ such that $x_i=y_i$ for all $i \in I$.
This $x$ satisfies that $x_i = \maj_p(x^{(1)}_i,x^{(2)}_i,x^{(3)}_i)$ for all $i \in I$ such that $\maj_p(x^{(1)}_i,x^{(2)}_i,x^{(3)}_i)$ is defined.
\end{proof}

We note that none of the properties in \cref{prop:2-decompo-related-property} other than (iv) can characterize (iv), i.e., 2-decomposability.
We observe in the following example that it is indeed not possible to characterize 2-decomposability by closedness under operations by showing that 2-decomposability is not preserved by projection.
Note that closedness under operations is preserved by projection, which can be shown in a similar way as \cref{lemma:proj-closed} (i).
\begin{example}
Let $S=\{ (0,0,1,2), (0,1,0,3), (1,0,0,4) \}$.
Then $S$ is 2-decomposable.
To see this, take $x \in \Join_{1\le i<j\le 4}\proj_{i,j}(S)$.
First, assume that $\proj_{1,4}(x)=(0,2)$.
As $\proj_{1,2}(x) \in \proj_{1,2}(S)$ and $\proj_{1,3}(x) \in \proj_{1,3}(S)$, we must have $x = (0,0,1,2)$.
Hence, we have $x \in S$.
We can similarly show $x \in S$ in the cases of $\proj_{1,4}(x)=(0,3)$ and $\proj_{1,4}(x)=(1,4)$.
Hence, $S$ is 2-decomposable.
On the other hand, $S' := \proj_{1,2,3}(S)$, which equals $\{ (0,0,1), (0,1,0), (1,0,0) \}$,
is not 2-decomposable, since $(0,0,0) \in \Join_{1\le i<j\le 3}\proj_{i,j}(S')\setminus S'$.
Hence, 2-decomposability is not preserved by projection in general.
\end{example}

While 2-decomposability cannot be characterized by closedness under operations, we show that it can be indeed characterized by weak closedness under infinitely many partial operations.
Intuitively, we prepare all the partial operations that correspond to the condition of 2-decomposability.
Concretely, assume that $S \subseteq \mathbb{Z}^4$ is 2-decomposable.
Then $S$ contains all vectors $x=(x_1,x_2,x_3,x_4)$ such that $\proj_{i,j}(x) \in \proj_{i,j}(S)$ for all $1 \le i < j \le 4$.
The latter condition can be rewritten as $(x_1,x_2,*,*)$, $(x_1,*,x_3,*)$, $(x_1,*,*,x_4)$, $(*,x_2,x_3,*)$, $(*,x_2,*,x_4)$, $(*,*,x_3,x_4)\in S$, where $*$ can be any (distinct) integer.
Then the condition $x \in S$ is guaranteed by weak closedness under the partial operation $f:\mathbb{Z}^6 \to \mathbb{Z}$ such that $f(y_1,y_1,y_1,*,*,*)=y_1$, $f(y_2,*,*,y_2,y_2,*)=y_2$, $f(*,y_3,*,y_3,*,y_3)=y_3$, and 
$f(*,*,y_4,*,y_4,y_4)=y_4$ for all $y_1,y_2,y_3,y_4 \in \bb{Z}$, where $*$ can be any (distinct) integer, and $f$ is undefined for the other inputs.
In this way, we prepare partial operations 
that correspond to each dimension to characterize 2-decomposability by weak closedness under operations.

Now, we formally state the above intuition.
For each integer $k\ge 2$, we will define a partial operation $f^{(k)}$ over $\bb{Z}$ with $\dom (f^{(k)}) \be \bb{Z}^{\binom{k}{2}}$.
We first define $\dom(f^{(k)})$.
Let $T_k$ be the set of two-element subsets of $[k]:=\{1,\dots,k\}$, that is
\[T_k:=\{\{\ell,m\}\be  [k] \mid \ell \neq m\}.\]
Thus, we have 
$|T_k|=\binom{k}{2}$. 
We identify $\bb{Z}^{T_k}$ with $\bb{Z}^{\binom{k}{2}}$ below.
Then, we define a family $\mathscr{S}_k\subseteq {\rm Pow}(T_k)$ as 
$\mathscr{S}_k:=\{S_1,\dots,S_k\}$ 
where each $S_{\ell}$ is defined as 
\[S_{\ell}:=\{\{p,q\}\in T_k~|~\ell \in \{p,q\}\}.\]
By definition, we have $|S_{\ell}| =k-1$, and 
for different $\ell,m$, we have
$S_{\ell}\cap S_m=\{\{\ell,m\}\}$. 

Then, for each $d\in \bb{Z}$ and $S_{\ell}\in \mathscr{S}_k$, 
we define $Q_d(S_{\ell})\subseteq \bb{Z}^{T_k}$ as 
\[Q_d(S_{\ell}):=\{(x_{\{p,q\}})_{\{p,q\} \in T_k}\in \bb{Z}^{T_k} \mid x_{\{\ell,m\}}=d  ~{\rm for}~ m\in [k]\setminus \{\ell\}\}.\]

Now, taking union of all of the above set for $d\in \bb{Z}, S_{\ell}\in \mathscr{S}_k$, 
we obtain $\dom(f^{(k)}) \be \bb{Z}^{T_k}$, that is
\[\dom(f^{(k)}):=\bigcup_{d\in \bb{Z}, S_{\ell}\in \mathscr{S}_k}Q_d(S_{\ell}).\]
Note that if 
$(x_{\{p,q\}})_{\{p,q\} \in T_k}\in Q_d(S_{\ell})\cap Q_{d'}(S_m)$, 
we have
$d=x_{\{\ell,m\}}=d'$. 
Therefore, for $(x_{\{p,q\}})_{\{p,q\} \in T_k}\in \dom(f^{(k)})$, 
we define $f^{(k)}((x_{\{p,q\}})_{\{p,q\} \in T_k}):=d$ 
where $d\in \bb{Z}$ is the unique element such that $(x_{\{p,q\}})_{\{p,q\} \in T_k}\in Q_d(S_{\ell})$ for some $\ell$.

Then we define a set $F$ of partial operations as
\[F:=\{f^{(k)} \mid k\ge 2\}. \]

Now we are ready to show our result that characterizes 2-decomposability.
\begin{theorem}
\label{thm:2-decompo<=>F-closed}
Let $S$ be a subset of $\mathbb{Z}^n$.
$S$ is 2-decomposable if and only if $S$ is weakly $f$-closed for each $f \in F$.
\end{theorem}

\begin{proof}
For the if part, 
take any $x \in \Join_{i,j}\pi_{i,j}(S)$.
It suffices to show that $x \in S$.
For each $1 \le i < j \le n$, 
there exists $s^{(i,j)}\in S$ such that $\pi_{i,j}(x)=\pi_{i,j}(s^{(i,j)})$.
Now, applying $f^{(n)}$ to vectors $s^{(i,j)}$ $(1 \le i < j \le n)$, we obtain $x$.
Since $S$ is weakly $f^{(n)}$-closed, we have $x \in S$. 
Therefore, $S$ is 2-decomposable.

We then show the only-if part.
Assume that $S$ is 2-decomposable.
We show that $S$ is weakly $f^{(k)}$-closed for each $f^{(k)} \in F$.
Fix $k \ge 2$ and take any (multi-)set $\{s^{(\ell,m)} \mid \{\ell,m\}\in T_k \}$ of elements of $S$ on which $f^{(k)}$ is defined.
For $1\le i \le n$, define $r^{(i)} \in \bb{Z}^{T_k}$ as 
\[r^{(i)}_{\{\ell,m\}}:=s^{(\ell,m)}_i.\]
Then we have $r^{(i)} \in \dom(f^{(k)})$ for each $i$.
Moreover, by definition, 
we have $f^{(k)}((s^{(\ell,m)})_{\{\ell,m\} \in T_k})=(f^{(k)}(r^{(i)}))_{i \in [n]}$. 
Hence, it suffices that $(f^{(k)}(r^{(i)}))_{i \in [n]} \in S$.
For each $i\in [n]$, arbitrarily take $d_i\in \bb{Z}$ and $\ell_i \in [k]$ such that $r^{(i)} \in Q_{d_i}(S_{\ell_i})$.
Then we have $\pi_{i,j}((f^{(k)}(r^{(i)}))_{i \in [n]})
=(f^{(k)}(r^{(i)})),f^{(k)}(r^{(j)}))
=(d_i,d_j)$.
On the other hand, when $\ell_i \neq \ell_j$, we have $s^{(\ell_i,\ell_j)}_i=r^{(i)}_{\{\ell_i,\ell_j\}} = 
d_i$ and $s^{(\ell_i,\ell_j)}_j=r^{(j)}_{\{\ell_i,\ell_j\}} = d_j$, implying that $\pi_{i,j}(s^{(\ell_i,\ell_j)})=(d_i,d_j) \in \pi_{i,j} (S)$.
Moreover, when $\ell_i = \ell_j$, by taking arbitrary $\ell \in [n] \setminus \{ \ell_i \}$, 
we have $s^{(\ell_i,\ell)}_i=r^{(i)}_{\{\ell_i,\ell\}} = 
d_i$ and $s^{(\ell_i,\ell)}_j=r^{(j)}_{\{\ell_i,\ell\}} = d_j$, implying that $\pi_{i,j}(s^{(\ell_i,\ell)})=(d_i,d_j) \in \pi_{i,j} (S)$.
Therefore, $\pi_{i,j}((f^{(k)}(r^{(i)}))_{i \in [n]}) \in \pi_{i,j} (S)$ for each $i,j$, and 
since $S$ is 2-decomposable, we have $(f^{(k)}(r^{(i)}))_{i \in [n]} \in S$.
This completes the proof.
\end{proof}

\begin{remark}
\cref{thm:2-decompo<=>F-closed} is true even if $\bb{Z}$ is replaced with any set, since we do not use the property of $\bb{Z}$ in the proof.
\end{remark}

\section{Conclusion and future work}
\label{sec:conclusion}
We have analyzed the relationship between polyhedral representations of a set of integer vectors and closedness under operations in terms of 2-decomposability.
We especially show that the set of integer vectors is representable by a UTVPI system if and only if it is closed under the median operation and the directed discrete midpoint operation.
We also characterize 2-decomposability by weak closedness under partial operations.

Investigating other relationship between polyhedral representations and closedness under operations seems an interesting future direction. 
One candidate would be the relationship between Horn polyhedra (i.e., each inequality has at most one positive coefficient in the defining system) and the minimum operation.
It is known that the set of integer vectors in a Horn polyhedron is closed under the minimum operation~(see, e.g., \cite{MaD02}), but it is not known whether the converse is true.

\section*{Acknowledgments}
The authors thank Akihisa Tamura for helpful comments and suggestions to improve presentation of the paper.
The first author was partially supported by JST, ACT-X Grant Number JPMJAX200C, Japan, and JSPS KAKENHI Grant Number JP21K17700.
The third author was supported by JST CREST Grant Number JPMJCR22M1.
The fourth author was supported by WISE program (MEXT) at Kyushu University.

\bibliography{UTVPI}
\bibliographystyle{plain}
\section*{Appendix}

\subsection*{A More Direct Proof of \cref{lem:two-dim} in \cref{subsec:remaining-proof}}
Our proof of \cref{lem:two-dim} in \cref{subsec:remaining-proof} heavily relies on previous work, but we provide a more direct proof here for readability of the paper.

\newtheorem*{Lem-two-dim}{\cref{lem:two-dim}}
\begin{Lem-two-dim}
Let $S$ be a subset of $\mathbf{Z}^2$.
If $S$ is $\avgd$-closed, then $S$ is representable by a UTVPI system.
\end{Lem-two-dim}

\subsubsection*{Basic observations}
\label{subsubsec:basic-observations}

\begin{lemma}
\label{lemma:refl-UTVPI}
Let $S$ be a subset of $\mathbb{Z}^2$.
If $S$ is representable by a UTVPI system, then so are the following sets:
\begin{itemize}
 \item[(i)] the reflection of $S$ over the $x_1$-axis, i.e., $\{ (x_1,-x_2) \mid (x_1,x_2) \in S \}$, 
 \item[(ii)] the reflection of $S$ over the $x_2$-axis, i.e., $\{ (-x_1,x_2) \mid (x_1,x_2) \in S \}$, 
 \item[(iii)] the reflection of $S$ over the line $x_1=x_2$, i.e., $\{ (x_2,x_1) \mid (x_1,x_2) \in S \}$, 
 \item[(iv)] the reflection of $S$ over the line $x_1=-x_2$, i.e., $\{ (-x_2,-x_1) \mid (x_1,x_2) \in S \}$, and 
 \item[(v)] the translation of $S$ by an integer vector $(a_1,a_2)$, i.e., $\{ (x_1+a_1,x_2+a_2) \mid (x_1,x_2) \in S \}$.
\end{itemize}
\end{lemma}

\begin{proof}
Since $S$ is represented by a UTVPI system, there exist a unit quadratic matrix $A$ and a vector $b$ such that $S = \{ x \in \mathbb{Z}^2 \mid Ax\ge b\}$.
We use the fact that for a surjection $f:\mathbb{Z}^2\to \mathbb{Z}^2$, 
if there exist a matrix $A'$ and a vector $b'$ such that $Ax\ge b$ if and only if $A'f(x)\ge b'$ for $x\in \mathbb{Z}^2$, 
then the image $f(S)$ is equal to 
$\{x \in \mathbb{Z}^2 \mid A'x\ge b'\}$.
In fact, for any $y \in f(S)$ there exists an element $x \in S$ such that $f(x)=y$, and since $Ax \ge b$ if and only if $A'y \ge b'$, we have $y \in \{x\in \mathbb{Z}^2 \mid A'x\ge b'\}$. 
Conversely, for any $y \in \{x\in \mathbb{Z}^2 \mid A'x\ge b'\}$, since $f$ is surjective there exists an $x \in \mathbb{Z}^2$ such that $f(x)=y$.
Now, since $Ax \ge b$ if and only if $A'y \ge b'$, this $x$ is in $S$ which means $y \in f(S)$. 
Using this fact, we show the lemma as follows.

\begin{itemize}
    \item [(i)]
    Let $A_1$ be a unit quadratic matrix whose the first column is the same as that of $A$ and the second column is $(-1)$ times that of $A$. 
    Then clearly $A \binom{x_1}{x_2}\ge b$ if and only if $A_1 \binom{x_1}{-x_2}\ge b$, and thus, $\{x\in \mathbb{Z}^2 \mid A_1x\ge b\}$ is a UTVPI representation of the reflection of $S$ over the $x_1$-axis.
    \item[(ii)]
    Let $A_2$ be a unit quadratic matrix whose the first column is $(-1)$ times that of $A$ and the second column is the same as that of $A$. 
    Then clearly $A \binom{x_1}{x_2}\ge b$ if and only if $A_2 \binom{-x_1}{x_2}\ge b$,
    and thus, $\{x \in \mathbb{Z}^2 \mid A_2x \ge b\}$ is a UTVPI representation of the reflection of $S$ over the $x_2$-axis.
    \item[(iii)]
    Let $A_3$ be a unit quadratic matrix 
 obtained by permuting the first and second column of $A$.
    Then clearly $A \binom{x_1}{x_2} \ge b$ if and only if $A_3 \binom{x_2}{x_1}\ge b$, and thus, $\{x\in \mathbb{Z}^2 \mid A_3x\ge b\}$ is a UTVPI representation of the reflection of $S$ over the line $x_1=x_2$. 
    \item[(iv)]
    Since the reflection of $S$ over the line $x_1=-x_2$ can be obtained by a sequence of the reflection over the line $x_1=x_2$, the reflection over the $x_1$-axis, and the reflection over the $x_2$-axis, the reflection of $S$ over the line $x_1=-x_2$ is representable by a UTVPI system by (i)-(iii).
    \item[(v)]
    Since $A\binom{x_1}{x_2}\ge b$ if and only if $A\binom{x_1+a_1}{x_2+a_2}\ge b+A\binom{a_1}{a_2}$, 
    $\{x\in \mathbb{Z}^2 \mid Ax\ge b+A\binom{a_1}{a_2}\}$ is 
    a UTVPI representation of the translation of $S$ by $(a_1,a_2)$. 
\end{itemize}
\end{proof}

\begin{lemma}
\label{lemma:closed-avgd}
Let $S$ be a subset of $\mathbb{Z}^2$.
If $S$ is $\avgd$-closed, then so are the following sets:
\begin{itemize}
 \item[(i)] the reflection of $S$ over the $x_1$-axis, i.e., $\{ (x_1,-x_2) \mid (x_1,x_2) \in S \}$, 
 \item[(ii)] the reflection of $S$ over the $x_2$-axis, i.e., $\{ (-x_1,x_2) \mid (x_1,x_2) \in S \}$, 
 \item[(iii)] the reflection of $S$ over the line $x_1=x_2$, i.e., $\{ (x_2,x_1) \mid (x_1,x_2) \in S \}$, 
 \item[(iv)] the reflection of $S$ over the line $x_1=-x_2$, i.e., $\{ (-x_2,-x_1) \mid (x_1,x_2) \in S \}$, and 
 \item[(v)] the translation of $S$ by an integer vector $(a_1,a_2)$, i.e., $\{ (x_1+a_1,x_2+a_2) \mid (x_1,x_2) \in S \}$.
\end{itemize}
\end{lemma}

\begin{proof}
We use the easy-to-prove fact that  if  a map $f:\mathbb{Z}^2\to \mathbb{Z}^2$ satisfies $\avgd(f(x),f(y))=f(\avgd(x,y))$ for $x,y\in \mathbb{Z}^2$, 
then the image $f(S)$ is also $\avgd$-closed.
In fact, for all $x',y'\in f(S)$ there exist $x,y\in S$ such that $f(x)=x', f(y)=y'$. 
Since $S$ in $\avgd$-closed, $\avgd(x,y)\in S$, and thus, $\avgd(x',y')=\avgd(f(x),f(y))=f(\avgd(x,y))\in f(S)$. 

\begin{itemize}
 \item[(i)]
 Since $\lfloor -a\rfloor =-\lceil a\rceil$ for $a\in \mathbb{R}$, 
 $\avgd(-x,-y)=-\avgd(x,y)$ for $x,y\in \mathbb{Z}$. 
 Thus, we have $\avgd(\binom{x_1}{-x_2}, \binom{y_1}{-y_2})=\binom{\avgd(x_1,y_1)}{-\avgd(x_2,y_2)}$.
 Hence the reflection of $S$ over the $x_1$-axis is also $\avgd$-closed.
 \item[(ii)]
 This can be similarly proven as (i).
 \item[(iii)]
 This is  followed by $\avgd(\binom{x_2}{x_1}, \binom{y_2}{y_1})=\binom{\avgd(x_2,y_2)}{\avgd(x_1,y_1)}$.
 \item[(iv)]
 The reflection of $S$ over the $x_1=-x_2$ is a composition of the reflection of $S$ over the $x_1$-axis, the $x_2$-axis, and the line $x_1=x_2$. Thus, it follows from (i)-(iii).
 \item[(v)]
 It follows from $\avgd(\binom{x_1+a_1}{x_2+a_2}, \binom{y_1+a_1}{y_2+a_2})=\avgd(\binom{x_1}{x_2}, \binom{y_1}{y_2})+\binom{a_1}{a_2}$. 
\end{itemize}
\end{proof}

\subsubsection*{Another proof of \cref{lem:two-dim}}

Given the basic observations in \cref{subsubsec:basic-observations}, 
we are ready to prove \cref{lem:two-dim} in a more self-contained manner.

\begin{proof}[Proof of \cref{lem:two-dim}]
Assume that $S$ is $\avgd$-closed.
By \cref{prop:two-dim-IntConv=UTVPI}, it suffices to show that $S$ is integrally convex.
Let $x\in \cl_{\conv}(S)$.
We show $x\in \cl_{\conv}(S\cap N(x))$ in the following.
By Carathéodory's theorem there exist three elements $t_1,t_2,t_3\in S$ such that $x\in \cl_{\conv}(\{t_1,t_2,t_3\})$. 

Now we will show the $\avgd$-closure $\cl_{\avgd}(\{t_1,t_2,t_3\})$ of $\{t_1,t_2,t_3\}$ is equal to the UTVPI hull $\cl_{\utvpi}(\{t_1,t_2,t_3\})$ of it.
First we show a similar equality for  two-point case, i.e., $\cl_{\avgd}(\{t_1,t_2\})=\cl_{\utvpi}(\{t_1,t_2\})$. 
By \cref{cor:avg-in-UTVPI}, 
$\cl_{\avgd}(\{t_1,t_2\})$ is included in $\cl_{\utvpi}(\{t_1,t_2\})$. 
If a line  $\ell$ through  $t_1, t_2$ is parallel to $x_1$-axis, $x_2$-axis, the line $x_1=x_2$, or the line $x_1=-x_2$, it is clear that $\cl_{\avgd}(\{t_1,t_2\})$ consists of  all integer points between $t_1$ and $t_2$ on $\ell$ and $\cl_{\utvpi}(\{t_1,t_2\})$ is also the same set, and thus, $\cl_{\avgd}(\{t_1,t_2\})=\cl_{\utvpi}(\{t_1,t_2\})$. 
If not, from
\cref{lemma:refl-UTVPI} and  \cref{lemma:closed-avgd},
by reflecting and translating $\cl_{\avgd}(\{t_1,t_2\})$ and $\cl_{\utvpi}(\{t_1,t_2\})$ appropriately, 
it suffices to show the case of $t_1=(0,0)$ and $t_2=(p,q)$ where $q>p>0$. 
In this case, $\cl_{\utvpi}(\{t_1,t_2\})$ consists integer points $(x_1,x_2) $ that satisfy 
\[0 \le x_1 \le p, x_1 \le x_2 \le x_1+q-p. \]
By $\cl_{\avgd}(\{t_1,t_2\})\subseteq \cl_{\utvpi}(\{t_1,t_2\})$, 
all points in $\cl_{\avgd}(\{t_1,t_2\})$  also satisfy these inequalities.

First  we will show $(0,q-p)\in \cl_{\avgd}(\{t_1,t_2\})$. 
If we show some point $(0,r)$ is in $\cl_{\avgd}(\{t_1,t_2\})$ where $0<r\le q-p$, 
we retake $t_1=(0,r), t_2=(p,q)$ and by induction on $q$, 
we have $(0,q-p)\in \cl_{\avgd}(\{t_1,t_2\})$. 
And if we show there exists $(\tilde{p},\tilde{q}) \in \cl_{\avgd}(\{t_1,t_2\})$ that satisfies $\tilde{q}>\frac{q}{p} \tilde{p}$, then, 
since $\tilde{p}<p$, by induction on $p$, we have $(0,r)\in \cl_{\avgd}(\{t_1,t_2\})$. 
Thus, we will show that such $(\tilde{p},\tilde{q})$ is in $\cl_{\avgd}(\{t_1,t_2\})$.
When $p,q$ are both even, $\avgd(p,q)=(\frac{p}{2}, \frac{q}{2})\in \cl_{\avgd}(\{t_1,t_2\})$ and 
by repeating halving the point, we have a $(p', q')\in \cl_{\avgd}(\{t_1,t_2\})$ such that at least  one of $p'$ or $q'$ is odd number. 
First we consider the case when $p'$ is odd. Then
$\avgd((0,0),(p',q'))=(\frac{p'}{2}-\frac{1}{2}, \lfloor \frac{q'}{2}\rfloor)$ 
and 
\[ \lfloor \frac{q'}{2}\rfloor\ge \frac{q'}{2}-\frac{1}{2}=\frac{q}{p}(\frac{p'}{2}-\frac{1}{2})+\frac{1}{2}(\frac{q}{p}-1)>\frac{q}{p}(\frac{p'}{2}-\frac{1}{2})\]
which means we obtain desirable  $(\tilde{p},\tilde{q})=(\frac{p'}{2}-\frac{1}{2}, \lfloor \frac{q'}{2}\rfloor)$.
Otherwise; $p'$ is even and $q'$ is odd, 
and 
$\avgd((p',q'),(0,0))=(\frac{p'}{2}, \frac{q'}{2}+\frac{1}{2})$. Here
\[\frac{q'}{2}+\frac{1}{2}>\frac{q}{p}\cdot \frac{p'}{2}\]
and we again obtain desirable  $(\tilde{p},\tilde{q})=
(\frac{p'}{2}, \frac{q'}{2}+\frac{1}{2})$.
Hence $(0,q-p)$ is in $\cl_{\avgd}(\{t_1,t_2\})$. 
Now, $(p,p)$ is also in $\cl_{\avgd}(\{t_1,t_2\})$, since
it is obtained by translating $(0,q-p)$ by $(-p,-q)$ and reflect it over $x_1$-axis and $x_2$-axis.
Thus, all vertices of $\cl_{\utvpi}(\{t_1,t_2\})$ is in $\cl_{\avgd}(\{t_1,t_2\})$.
Now, since the boundary of $\cl_{\utvpi}$ is parallel to 
$x_1$-axis, $x_2$-axis, the line $x_1=x_2$, or the line $x_1=-x_2$, 
all points on the boundaries are in $\cl_{\avgd}(\{t_1,t_2\})$.  
Finally, any inner point $(x_1,x_2)\in \cl_{\utvpi}(\{t_1,t_2\})$ lies  
on the line segment that connects 
$(x_1,x_1)$ and $(x_1,x_1+q-p)$
which are on the boundaries, 
implying that 
$(x_1,x_2)$ is in $\cl_{\avgd}(\{t_1,t_2\})$.  
Therefore, we have $\cl_{\avgd}(\{t_1,t_2\})=\cl_{\utvpi}(\{t_1,t_2\})$. 

Next, we show $\cl_{\avgd}(\{t_1,t_2,t_3\})=\cl_{\utvpi}(\{t_1,t_2,t_3\})$. 
By \cref{cor:avg-in-UTVPI}, 
$\cl_{\avgd}(\{t_1,t_2,t_3\})\subseteq \cl_{\utvpi}(\{t_1,t_2,t_3\})$, 
and by the same argument of the two points case, 
it is enough to show any vertex $v$ of $\cl_{\utvpi}(\{t_1,t_2,t_3\})$ is in $\cl_{\avgd}(\{t_1,t_2,t_3\})$. 
When $v$ coincides with one of $t_1,t_2,t_3$, $v$ is trivially in $\cl_{\avgd}(\{t_1,t_2,t_3\})$.
Thus, we assume $v$ equals none of $t_1,t_2,t_3$. 
If necessary, by retaking the indexes of $t_1,t_2,t_3$, $v$ is the intersection of boundary through $t_1$ and that through $t_2$. 
Let $\theta$ be the angle between the two boundaries. 
Then, $\theta$ is one of $\frac{1}{4}\pi, \frac{1}{2}\pi,$ or $\frac{3}{4}\pi$,  
but  we show that  $\theta$ must be $\frac{3}{4}\pi$. 
In fact, assume $\theta=\frac{1}{4}\pi$. 
By 
\cref{lemma:refl-UTVPI} and \cref{lemma:closed-avgd}, 
we may assume $v=(0,0)$, $t_1$ lies on the positive part of $x_1$-axis, and $t_2$ lies on the line $x_1=x_2$ and is in the first quadrant. 
Then, $t_1,t_2,t_3$ satisfy $x_1\ge 1$, and thus, $v=(0,0)$ is not in $\cl_{\utvpi}(\{t_1,t_2,t_3\})$, a contradiction. 
Moreover, let $\theta=\frac{1}{2}\pi$. 
Again we may assume (i) $v=(0,0)$, and $t_1$ and $t_2$ lie on the positive parts of $x_1$-axis and $x_2$-axis,  respectively or (ii) $v=(0,0)$, 
$t_1$  lies on the line $x_1=x_2$ and is in the first quadrant, and 
$t_2$  lies on the line $x_1=-x_2$ and is in the forth quadrant.   
In case (i)
 $t_1,t_2,t_3$ satisfy $x_1+x_2\ge 1$, 
in case (ii)
$t_1,t_2,t_3$ satisfy $x_1\ge 1$
, and thus, in both cases $v=(0,0)$ is not in $\cl_{\utvpi}(\{t_1,t_2,t_3\})$, a contradiction.
Hence $\theta=\frac{3}{4}\pi$. 
We assume $v=(0,0)$, $t_1=(a,0)$, and $t_2=(-b,b)$ where $a,b$ are positive integers. 
Then by the case of two points, $v=(0,0)\in \cl_{\avgd}(\{t_1,t_2\})\subseteq \cl_{\avgd}(\{t_1,t_2,t_3\})$. 
Therefore, all vertices of $\cl_{\utvpi}(\{t_1,t_2,t_3\})$ are in $\cl_{\avgd}(\{t_1,t_2,t_3\})$. 
Therefore, we have $\cl_{\avgd}(\{t_1,t_2,t_3\})=\cl_{\utvpi}(\{t_1,t_2,t_3\})$. 

Now $x\in \cl_{\conv}(\{t_1,t_2,t_3\})\subseteq \cl_{\conv}(\cl_{\avgd}(\{t_1,t_2,t_3\}))$. 
By $\cl_{\avgd}(\{t_1,t_2,t_3\})=\cl_{\utvpi}(\{t_1,t_2,t_3\})$ and \cref{prop:two-dim-IntConv=UTVPI}, $\cl_{\avgd}(\{t_1,t_2,t_3\})$ is integrally convex and thus
$x\in \cl_{\conv}(\cl_{\avgd}(\{t_1,t_2,t_3\})\cap N(x))$. 
Since $S$ is $\avgd$-closed, 
$\cl_{\avgd}(\{t_1,t_2,t_3\})\subseteq S$ and thus
$x\in \cl_{\conv}(S\cap N(x))$.
Therefore, $S$ is integrally convex.
\end{proof}

\end{document}